\documentclass[sigconf]{acmart}
\AtBeginDocument{%
  }

\setcopyright{acmlicensed}
\copyrightyear{2026}
\acmYear{2026}
\acmDOI{XXXXXXX.XXXXXXX}
\acmConference[Conference acronym 'XX]{Make sure to enter the correct
  conference title from your rights confirmation email}{June 03--05,
  2018}{Woodstock, NY}
\acmISBN{978-1-4503-XXXX-X/2018/06}

\usepackage{algorithm}
\usepackage{algorithmic}
\usepackage{multirow}
\usepackage{booktabs}
\usepackage{graphicx}
\usepackage{pifont}
\usepackage{amsmath}
\usepackage{bm}
\usepackage{enumitem}
\usepackage{subcaption}
\usepackage{xcolor}
\usepackage{balance}
\newtheorem{remark}{Remark}

\begin{document}

\title{From Transfer to Collaboration: A Federated Framework for Cross-Market Sequential Recommendation}

\author{Jundong Chen}
\affiliation{%
  \institution{Beijing Jiaotong University}
  \city{Beijing}
  \country{China}
}
\email{jundongchen@bjtu.edu.cn}

\author{Honglei Zhang}
\affiliation{%
  \institution{Beijing Jiaotong University}
  \city{Beijing}
  \country{China}
}
\email{honglei.zhang@bjtu.edu.cn}

\author{Xiangmou Qu}
\affiliation{%
  \institution{OPPO Research Institute}
  \city{Shenzhen}
  \country{China}
}
\email{lokinko.cs@gmail.com}

\author{Haoxuan Li}
\affiliation{%
 \institution{Peking University}
 \city{Beijing}
 \country{China}
}
\email{hxli@stu.pku.edu.cn}
 
\author{Han Yu}
\affiliation{%
  \institution{Nanyang Technological University}
  \country{Singapore}
}
\email{han.yu@ntu.edu.sg}

\author{Yidong Li}
\authornote{Corresponding author}
\affiliation{%
  \institution{Beijing Jiaotong University}
  \city{Beijing}
  \country{China}
}
\email{ydli@bjtu.edu.cn}

\renewcommand{\shortauthors}{Chen et al.}

\begin{abstract}
Cross-market recommendation (CMR) aims to enhance recommendation performance across multiple markets. Due to its inherent characteristics, \textit{i.e.}, data isolation, non-overlapping users, and market heterogeneity, CMR introduces unique challenges and fundamentally differs from cross-domain recommendation (CDR). Existing CMR approaches largely inherit CDR by adopting the \textbf{one-to-one transfer} paradigm, where a model is pretrained on a source market and then fine-tuned on a target market. However, such a paradigm suffers from \textbf{CH1. source degradation}, where the source market sacrifices its own performance for the target markets, and \textbf{CH2. negative transfer}, where market heterogeneity leads to suboptimal performance in target markets. To address these challenges, we propose FeCoSR, a novel federated collaboration framework for cross-market sequential recommendation. Specifically, to tackle \textbf{CH1}, we introduce a many-to-many collaboration paradigm that enables all markets to jointly participate in and benefit from training. It consists of a federated pretraining stage for capturing shared behavior-level patterns, followed by local fine-tuning for market-specific item-level preferences. For \textbf{CH2}, we theoretically and empirically show that vanilla Cross-Entropy (CE) exacerbates market heterogeneity, undermining federated optimization. To address this, we propose a Semantic Soft Cross-Entropy (S$^2$CE) that leverages shared semantic information to facilitate collaborative behavioral learning across markets. Then, we design a market-specific adaptation module during fine-tuning to capture local item preferences. Extensive experiments on the real-world datasets demonstrate the advantages of FeCoSR over other methods. Our code is available at https://github.com/jundongchen13/FeCoSR.
\end{abstract}




\maketitle

\section{INTRODUCTION}\label{sec:intro}
Recommender systems~\cite{liu2024multimodal, zhao2024recommender, hu2025modality} play a crucial role in suggesting items of interest to users, generating substantial commercial value for multimedia platforms such as TikTok, Spotify, and Amazon. As these companies continue to expand beyond their domestic markets, cross-market recommendation (CMR) has emerged to improve recommendation performance across multiple markets. Unlike cross-domain recommendation (CDR)~\cite{zang2022survey, zhu2022personalized}, CMR exhibits unique characteristics, including data isolation, non-overlapping users, and market preference heterogeneity, making it a distinct and increasingly important research problem~\cite{bonab2021cross, wang2024pre}. 

Early CMR methods, such as FOREC~\cite{bonab2021cross} and MA~\cite{bhargav2023market}, achieve promising results by collecting geographically isolated data for centralized training. However, the introduction of privacy regulations, \textit{e.g.}, the GDPR in Europe and the CCPA in the United States, renders such centralized approaches impractical. To enable privacy-preserving CMR, adopting the \textbf{one-to-one transfer} paradigm from CDR has become a prevalent solution. For example, CAT-SR~\cite{wang2024pre}, inspired by UniSRec~\cite{hou2022towards}, pretrains a transferable model on source markets to improve recommendation performance in target markets. Despite their effectiveness, we argue that such paradigms face two fundamental challenges: \textbf{CH1. Source degradation}. The source market is required to learn a market-agnostic model to ensure transferability, which can significantly sacrifice its own performance. As illustrated in Figure~\ref{pic:preexp}(a), the models pretrained by UniSRec and CAT-SR show degraded performance on the source market. Even after fine-tuning, they still fail to reach the performance of locally trained models. \textbf{CH2. Negative transfer}. Due to heterogeneous market preferences, pretraining on a single market fails to identify truly market-agnostic knowledge. As a result, source-specific preference patterns are inevitably transferred to the target market, which can even harm its performance.

\begin{figure}[t]
\centering
\includegraphics[width=1.0\columnwidth]{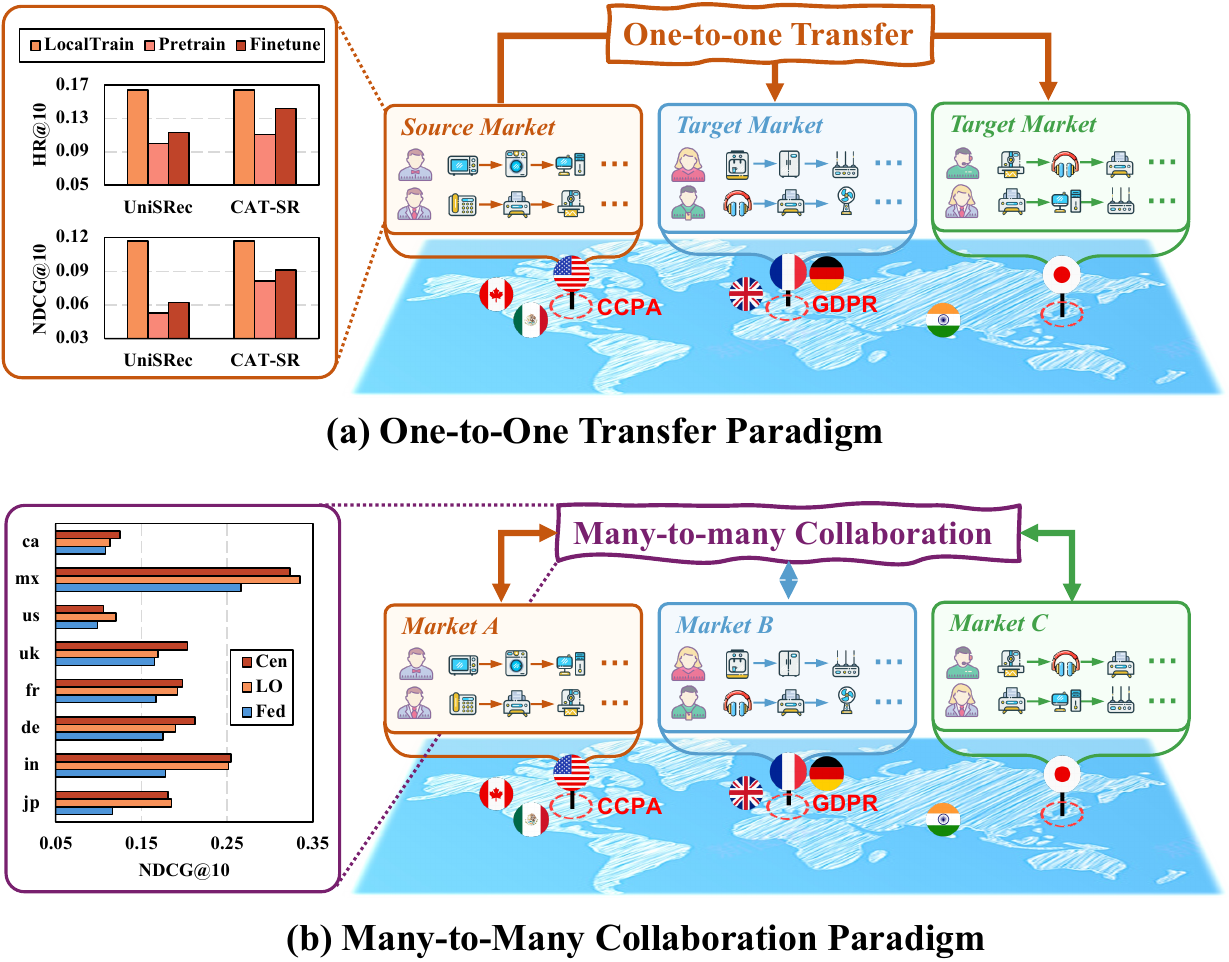} 
\caption{Paradigm comparison. One-to-one transfer suffers from performance degradation on the source market and negative transfer to the target market. Many-to-many collaboration is promising but not well explored by existing methods, motivating our federated collaboration framework.}
\label{pic:preexp}
\end{figure}

To overcome the limitations of the one-to-one transfer paradigm, we explore a \textbf{many-to-many collaboration} paradigm. Specifically, we conduct exploratory experiments comparing two representative collaborative strategies, centralized (Cen) and federated (Fed) learning with FedAvg~\cite{mcmahan2017communication}, along with local-only (LO) training. The results in Figure~\ref{pic:preexp}(b) reveal two key findings: i) Cen can partially exploit cross-market collaboration but is often impractical due to privacy constraints. Compared with LO, it improves some markets, \textit{e.g.}, mx, uk, while degrading others, \textit{e.g.}, us. ii) Although Fed preserves data privacy, the traditional Fed strategy fails to achieve effective collaboration under significant market heterogeneity. These observations motivate us to develop a privacy-preserving collaborative framework for CMR that tackles cross-market heterogeneity and enables effective knowledge sharing across markets.

In this paper, we focus on cross-market sequential recommendation (CMSR) and propose \textbf{FeCoSR}, a \underline{Fe}derated \underline{Co}llaboration framework for cross-market \underline{S}equential \underline{R}ecommendation. To address \textbf{CH1}, we adopt a two-stage strategy. First, all markets collaboratively pretrain transferable behavior-level patterns by leveraging the textual modality via federated optimization. Next, each market performs local fine-tuning to further incorporate the ID modality, capturing item-level preferences. This strategy enables all markets to actively participate and benefit, achieving \textbf{consistent improvements} across markets. For \textbf{CH2}, we theoretically and empirically reveal that the widely used one-hot supervision in vanilla Cross-Entropy (CE), \textit{i.e.}, treating all non–ground-truth items as negative samples, exacerbates market heterogeneity and disturbs federated optimization. To mitigate this, we propose a Semantic Soft Cross-Entropy (S$^2$CE) that replaces one-hot targets with a semantically smoothed supervision. By emphasizing item textual semantics, S$^2$CE encourages the learning of behavior-level patterns and alleviates heterogeneity, facilitating \textbf{robust sharing} across markets. After that, we employ a market-specific adaptation during fine-tuning to capture precise local item-level preferences.

In summary, the main contributions of this work are as follows:

\begin{itemize}[left=0pt]
\item We identify two fundamental challenges in the widely adopted one-to-one transfer paradigm for CMR, \textit{i.e.}, source degradation and negative transfer, and propose a novel two-stage \textbf{federated collaboration} framework to tackle both challenges.

\item In the federated pretraining stage, we theoretically analyze the source of cross-market heterogeneity and introduce a Semantic Soft Cross-Entropy objective to mitigate it, encouraging the learning of shared behavior-level patterns.

\item In the local fine-tuning stage, we design a market-specific adaptation module to capture local item-level preferences, preserving the uniqueness of each market.

\item Extensive experiments on real-world CMR datasets demonstrate that the proposed method effectively overcomes the limitations of existing approaches and achieves superior performance.
\end{itemize}

\section{RELATED WORK}
\subsection{Cross-market Recommendation}
Cross-market recommendation (CMR) has emerged as a novel problem distinct from cross-domain recommendation (CDR), both aiming to mitigate data sparsity. While CDR focuses on transferring knowledge across domains with different feature spaces, CMR emphasizes knowledge sharing across regions with heterogeneous preferences~\cite{bonab2021cross}. Due to the inherent geographic isolation of CMR, users are often non-overlapping, and data cannot be directly shared, making it a particularly challenging task~\cite{wang2024pre}. Existing CMR methods fall into three categories: i) \textbf{Meta-learning methods} such as FOREC~\cite{bonab2021cross}, MAML-CF~\cite{kang2023outlier}, and M$^3$Rec~\cite{cao2022item} pretrain models on multiple markets and fine-tune on target markets. They mainly learn global representations and overlook market-specific characteristics. ii) \textbf{Market-aware modeling methods}, including MA~\cite{bhargav2023market} and Bert4CMF~\cite{hu2024enhancing}, incorporate market-aware components to capture regional patterns. However, these methods above typically require full data sharing, which is often infeasible due to privacy or regulatory constraints. iii) \textbf{Transfer-based methods.} CAT-SR~\cite{wang2024pre} pretrains on a source market and transfers only privacy-insensitive model parameters for target-specific fine-tuning. As discussed in Section~\ref{sec:intro}, this one-to-one transfer paradigm suffers from unfairness and negative transfer. To address these issues, we propose a \textbf{federated collaboration} paradigm, which enables effective cross-market knowledge sharing while preserving user privacy.

\subsection{Federated Recommendation}
Federated recommendation (FR)~\cite{sun2024survey} is a privacy-preserving paradigm for personalized services, where a central challenge lies in handling heterogeneous preferences across clients. Existing approaches can be broadly categorized into two groups: i) \textbf{Cross-user FR}~\cite{zhang2023lightfr, zhang2026transfr, li2026fedau2, qian2025personalized, chen2025beyond}, where each user is treated as a client. From the server perspective, GPFedRec~\cite{zhang2024gpfedrec} performs client-specific aggregation based on model similarity, FedCA~\cite{zhang2026beyond} further incorporates data complementarity, and FedSC~\cite{gui2026federated} adopts stochastic aggregation to preserve diversity in global parameters. From the client perspective, PFedRec~\cite{zhang2023dual} employs dual personalization locally, FedRAP~\cite{li2023federated} augments the global model with an additional local module, and FedEA~\cite{chen2026breaking} selectively absorbs knowledge from the global model. ii) \textbf{Cross-silo FR}~\cite{chen2023win,guo2024prompt,wang2025federated}, where each client corresponds to a group of users. FR-JVE~\cite{li2025efficient} studies collaborative recommendation across platforms within a venture ecosystem. However, it requires overlapping users, which is not applicable to CMR. FedDCSR~\cite{zhang2024feddcsr} addresses feature distribution inconsistency in cross-domain recommendation via decoupled learning, whereas CMR focuses on market preference heterogeneity. Overall, existing FR methods are designed for specific settings and are not well-suited for CMR. By leveraging the privacy-preserving advantages of FR, our framework further enables effective multi-market collaborative training, thereby improving recommendation performance across markets.

\section{METHODOLOGY}\label{sec:methodology}
\subsection{Preliminaries}
\subsubsection{Cross-Market Sequential Recommendation} Let $\mathcal{M}$ denote the set of markets, where each market $m \in \mathcal{M}$ contains a set of users $\mathcal{U}^m$ and items $\mathcal{I}^m$. For each user $u \in \mathcal{U}^m$, we denote the historical interaction sequence as $s_u^m = [i_1^m, i_2^m, \dots, i_{|s_u^m|}^m]$, where $i_k^m \in \mathcal{I}^m$. The dataset is defined as $\mathcal{S}^m = \{ s_u^m \}_{u \in \mathcal{U}^m}$. For each market $m$, the sequential recommendation model is defined as $\Theta^m = \{\mathbf{E}^m, f_\theta^m(\cdot)\}$, where $\mathbf{E}^m \in \mathbb{R}^{|\mathcal{I}^m| \times d}$ denotes the item embedding matrix, with $|\mathcal{I}^m|$ being the number of items and $d$ the embedding dimension. The behavior encoder $f_\theta^m(\cdot)$ maps user's interaction sequence to a representation $\mathbf{h}_u^m = f_\theta^m(s_u^m) \in \mathbb{R}^d$. The market $m$ then computes the scores $\mathbf{r}_u^m = \mathbf{h}_u^m (\mathbf{E}^m)^\top \in \mathbb{R}^{|\mathcal{I}^m|}$ to recommend top-k items for each user $u$. The goal of CMR is to leverage knowledge across markets to learn optimal models $\{\Theta^m\}_{m \in \mathcal{M}}$ for all markets.

\subsubsection{One-to-One Transfer} Due to privacy constraints and non-overlapping users across markets, existing CMR methods typically adopt the one-to-one transfer paradigm, which pretrains a transferable model $\Theta^{m_S}_{\text{pre}}$ on a source market $m_S$ and adapts the learned parameters to a target market $m_T$ by market-specific fine-tuning. The process can be formalized as:
\begin{equation}
    \Theta^{m_S}_{\text{pre}} = \arg\min_{\Theta} \mathcal{L}_{\text{pre}}(\Theta; \mathcal{S}^{m_S}),
\end{equation}
\begin{equation}
    \Theta^{m_T} = \arg\min_{\Theta} \mathcal{L}_{\text{ft}}(\Theta^{m_S}_{\text{pre}}; \mathcal{S}^{m_T}),
\end{equation}
where $\mathcal{L}_{\text{pre}}$ and $\mathcal{L}_{\text{ft}}$ denote the pretraining and fine-tuning objectives on the source and target markets, respectively. Representative methods such as UniSRec~\cite{hou2022towards} and CAT-SR~\cite{wang2024pre} adopt fixed semantic item embeddings as $\mathbf{E}$ and focus on learning transferable behavior encoders $f_\theta(\cdot)$ to facilitate cross-market transfer. However, as analyzed in Section~\ref{sec:intro}, such one-to-one transfer methods suffer from source degradation and negative transfer.

\subsubsection{Federated Collaboration} 
To promote consistent improvements and robust sharing, we propose a federated collaboration framework, consisting of federated pretraining across markets and local fine-tuning for each market. Formally, we have:
\begin{equation}\label{eq:fed_col}
    \Theta^{\text{fed}}_{\text{pre}} = \arg\min_{\Theta} \sum_{m \in \mathcal{M}} \alpha_m \, \mathcal{L}_{\text{pre}}(\Theta; \mathcal{S}^m),
\end{equation}
\begin{equation}\label{eq:local_ft}
    \Theta^m = \arg\min_{\Theta} \mathcal{L}_{\text{ft}}(\Theta^{\text{fed}}_{\text{pre}}; \mathcal{S}^m), \,\ \forall m \in \mathcal{M},
\end{equation}
where $\alpha_m$ is the weight for market $m$ during federated optimization. Eq.~(\ref{eq:fed_col}) enables all markets to collaboratively learn and benefit from shared behavior-level patterns while preserving data privacy. Eq.~(\ref{eq:local_ft}) ensures that each market can further capture its market-specific item-level preferences.

\subsection{Overview}
In this paper, we consider a company–market architecture where the company serves as the central server and each market acts as a client. The overall workflow of the framework is illustrated in Figure~\ref{pic:framework}. Our framework consists of two stages: the first stage leverages textual modality to learn shared behavior-level patterns via federated collaborative training, while the second stage introduces ID modality to capture market-specific item-level preferences.\\
\textbf{I. Federated  Collaborative Pretraining}. (1) \textit{Global Initialization.} The server first initializes global semantic item embeddings $\mathbf{E}$ using a pretrained text encoder over item descriptions, providing a shared semantic space across markets for behavioral pattern learning. (2) \textit{Local Pretraining.} With $\mathbf{E}^m=\mathbf{E}[\mathcal{I}^m]$ fixed, each market trains its local sequential encoder $f_\theta^m(\cdot)$ using \textit{semantic soft cross-entropy (S$^2$CE)} objective to learn behavior-level patterns. (3) \textit{Global Aggregation.} The local behavior encoders $\{f_\theta^m(\cdot)\}_{m\in\mathcal{M}}$ are uploaded to the server and aggregated to obtain a global encoder $f_\theta^{agg}(\cdot)$, which is redistributed to markets to replace $f_\theta^m(\cdot)$ for the next round.\\
\textbf{II. Local Market-specific Fine-tuning}. (4) After $R$ rounds, the global encoder $f^{agg}_\theta(\cdot)$ is fixed at each market. Each market then performs low-rank fine-tuning with an ID adapter to incorporate market-specific item-level preferences, optimized by the vanilla cross-entropy (CE) objective.
\begin{figure*}[ht]
\centering
\includegraphics[width=0.92\textwidth]{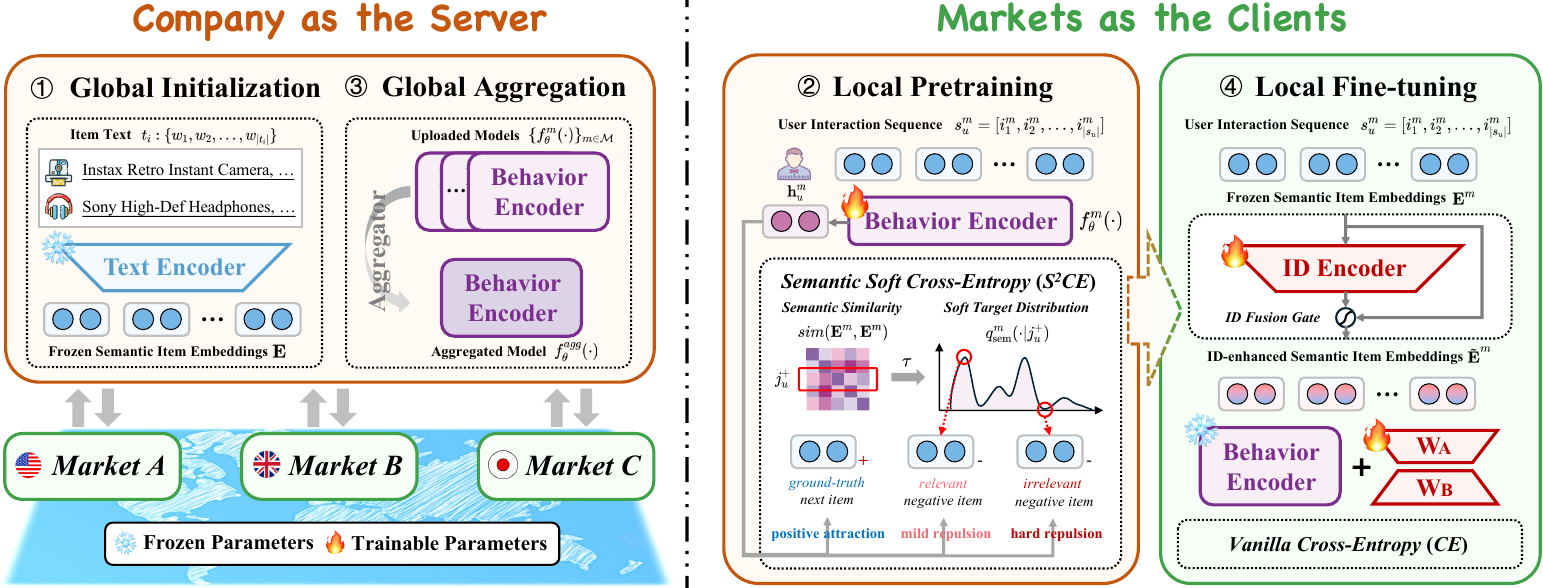} 
\caption{Overview of FeCoSR for CMR. I. (Orange) Federated pretraining with textual modality for cross-market behavior-level pattern learning; II. (Green) Local fine-tuning with ID modality for market-specific item-level preference modeling.}
\label{pic:framework}
\end{figure*}

\subsection{Federated Collaborative Pretraining}
\subsubsection{Global Initialization} To facilitate behavioral patterns sharing across markets, we leverage the textual modality to construct a shared representation space across markets~\cite{hou2022towards, wang2024pre}. Specifically, we utilize a pretrained text encoder, \textit{e.g.}, BERT~\cite{devlin2019bert}, to encode item textual descriptions into dense semantic vectors. For each item $i$, its textual content is denoted as $t_i=\{w_1,w_2,\dots,w_{|t_i|}\}$, where $w_k$ is the $k$-th word. The item embedding is obtained as:
\begin{equation}\label{eq:initialization}
    \mathbf{e}_i = \text{TextEncoder}([[\text{CLS}]; w_1, w_2, \dots, w_{|t_i|}]) \in \mathbb{R}^d,
\end{equation}
where the final hidden vector corresponding the first token $[\text{CLS}]$ is taken as the item embedding, and “[;]” denotes the concatenation operation. The server then broadcasts the semantic item embeddings $\mathbf{E} = \{\mathbf{e}_i\}_{i \in \mathcal{I}^{m_1} \cup \cdots \cup \mathcal{I}^{|\mathcal{M}|}}$ 
to all markets. During federated pretraining, each client keeps $\mathbf{E}^m = \mathbf{E}[\mathcal{I}^m]$ frozen and focuses on collaboratively optimizing the behavior encoder $f_\theta^m(\cdot)$.

\subsubsection{Local Pretraining} Generally, local sequential recommendation models are trained with the vanilla cross-entropy (CE) objective. This loss adopts one-hot supervision, where the ground-truth next item ${j_u^+}$ is assigned probability 1, while all other items are treated as negatives with probability 0. Formally, the CE loss is defined as:
\begin{equation}\label{eq:ce}
\mathcal{L}_{\text{CE}} = - \sum_{u \in \mathcal{U}^m} \log p_\theta^m(j_u^+|s_u^m),
\end{equation}
\begin{equation}
p_\theta^m(j|s_u^m) = 
\frac{\exp(\mathbf{h}_u^m \cdot \mathbf{e}^m_j)}
{\sum_{k \in \mathcal{I}^m} \exp(\mathbf{h}_u^m \cdot \mathbf{e}^m_k)},
\end{equation}
where $p_\theta^m(j|s_u^m)$ denotes the probability of item $j$ for user $u$ in market $m$. This formulation encourages the user representation $\mathbf{h}_u^m = f_\theta^m(s_u^m)$ to be pulled toward the ground-truth item $\mathbf{e}_{j_u^+}$ while being pushed away from other items, thereby enabling accurate next-item prediction. However, this training strategy overly emphasizes the ground-truth item, causing the model to focus on precise local preferences. While similar behavioral patterns, \textit{e.g.}, purchasing a charger after buying a phone, exist across markets, the specific item preferences differ, \textit{e.g.}, \textit{Apple} in market A versus \textit{Samsung} in market B~\cite{mcauley2015image}. Under vanilla CE, this results in heterogeneity across markets, which disturbs effective collaboration under federated optimization, as formalized in Proposition~\ref{pro:ib_vce}.

\begin{proposition}[Market Heterogeneity under Vanilla CE]\label{pro:ib_vce}
In cross-market recommendation, shared behavior-level patterns often exist across markets~\cite{mcauley2015image}. However, local training with the vanilla CE objective drives each market toward one-hot predictions on its item-level preferences, amplifying heterogeneity across markets.
\end{proposition}

\begin{proof}
Under vanilla CE optimization, the model maximizes the likelihood of the ground-truth item, that is:
\begin{equation}
    \min_\theta -\log p_\theta^m(j_u^+|s_u^m) \;\Longleftrightarrow\; \max_\theta p_\theta^m(j_u^+|s_u^m).
\end{equation}
Consider two users in markets $m_1$ and $m_2$ with ground-truth items $a$ and $b$. When the model approaches optimality, the predicted probabilities satisfy:
\begin{equation}\label{eq:p1_pos}
    p_\theta^{m_1}(a|s_u^{m_1})\ge1-\epsilon, \,\,\ p_\theta^{m_2}(b|s_u^{m_2})\ge1-\epsilon,
\end{equation}
with $\epsilon\to 0$. Thus, the remaining probabilities are bounded by:
\begin{equation}\label{eq:p1_neg}
    \sum_{j\neq a}p_\theta^{m_1}(j|s_u^{m_1})\le\epsilon,\,\,\ p_\theta^{m_2}(a|s_u^{m_2})\le\epsilon.
\end{equation}
The KL divergence between the two predictive distributions is:
\begin{equation}
    \mathrm{KL}(p_\theta^{m_1}\|p_\theta^{m_2})=\sum_j p_\theta^{m_1}(j|s_u^{m_1})\log\frac{p_\theta^{m_1}(j|s_u^{m_1})}{p_\theta^{m_2}(j|s_u^{m_2})}\\
\end{equation}
\begin{equation}\label{eq:kl_ce}
    =p_\theta^{m_1}(a|s_u^{m_1})\log\frac{p_\theta^{m_1}(a|s_u^{m_1})}{p_\theta^{m_2}(a|s_u^{m_2})}+\sum_{j\neq a}p_\theta^{m_1}(j|s_u^{m_1})\log\frac{p_\theta^{m_1}(j|s_u^{m_1})}{p_\theta^{m_2}(j|s_u^{m_2})}.
\end{equation}
Using Eqs.~(\ref{eq:p1_pos})–(\ref{eq:p1_neg}), the first term satisfies:
\begin{equation}
    p_\theta^{m_1}(a|s_u^{m_1})\log\frac{p_\theta^{m_1}(a|s_u^{m_1})}{p_\theta^{m_2}(a|s_u^{m_2})}\ge(1-\epsilon)\log\frac{1-\epsilon}{\epsilon},
\end{equation}
and the remaining terms are bounded by $O(\epsilon)$. Thus:
\begin{equation}
    \mathrm{KL}(p_\theta^{m_1}\|p_\theta^{m_2})\ge(1-\epsilon)\log\frac{1-\epsilon}{\epsilon}+O(\epsilon)\approx\log(1/\epsilon),
\end{equation}
showing that as $\epsilon\to0$, the predictive distributions diverge.

This example illustrates that item-level one-hot supervision induces substantial distributional differences across markets. Such heterogeneity makes it difficult to achieve effective collaboration through federated optimization~\cite{li2019convergence, li2020federated}.
\end{proof}

To mitigate the heterogeneity induced by vanilla CE, we propose a \textit{\textbf{Semantic Soft Cross-Entropy (S$^2$CE)}} loss, which converts the supervision from item-level to behavior-level learning. It encourages the model to respect items that are semantically close to the ground-truth item, while hardly pushing away dissimilar negative samples. Formally, the S$^2$CE loss is defined as:
\begin{equation}\label{eq:s2ce}
\mathcal{L}_{\text{S$^2$CE}}=- \sum_{u \in \mathcal{U}^m}\sum_{j \in \mathcal{I}^m}q^m_\text{sem}(j|j_u^+) \log p_\theta^m(j|s_u^m),
\end{equation}
where $q_\text{sem}(j|j_u^+)$ is a soft target distribution over all items in market $m$, reflecting their semantic similarity to the ground-truth next item $j_u^+$. Specifically, we define:
\begin{equation}
q^m_\text{sem}(j|j_u^+) = \frac{\exp(sim(\mathbf{e}_j, \mathbf{e}_{j_u^+}) / \tau)}{\sum_{k \in \mathcal{I}^m} \exp(sim(\mathbf{e}_k, \mathbf{e}_{j_u^+}) / \tau)},
\end{equation}
where $sim(\mathbf{e}_j, \mathbf{e}_{j_u^+})$ denotes the cosine similarity between the frozen semantic item embeddings, and $\tau$ is a temperature hyperparameter that controls the softness of the distribution.

\begin{proposition}[Heterogeneity Mitigation under S$^2$CE]
The S$^2$CE loss encourages the model to tolerate items that are semantically similar to the ground-truth. This supervision promotes the learning of shared behavior-level patterns and effectively mitigates the heterogeneity caused by item-level preferences across markets.
\end{proposition}
\begin{proof}
Following the same setting and notation as in the previous proof, we denote the semantic similarity between items $a$ and $b$ as $\sigma=sim(\mathbf e_a,\mathbf e_b)\in[-1,1]$. Under the S$^2$CE loss the soft target probability assigned to item $a$ in market $m_2$ is:
\begin{equation}\label{eq:sem_s2ce}
\begin{split}
q_\text{sem}^{m_2}(a|b) & = \frac{\exp(sim(\mathbf e_a,\mathbf e_b) / \tau)}{\sum_{k \neq b} \exp(sim(\mathbf{e}_k^\top \mathbf{e}_b) / \tau) + \exp(sim(\mathbf{e}_b^\top \mathbf{e}_b) / \tau)} \\
& \ge \frac{\exp(\sigma / \tau)}{|\mathcal{I}^{m_2}| \cdot \exp(1 / \tau)} 
= \frac{\exp({\sigma-1}/{\tau})}{|\mathcal{I}^{m_2}|}.
\end{split}
\end{equation}
Thus, the soft target assigns a nonzero probability mass to item $a$ in market $m_2$, which depends on the semantic similarity $\sigma$ and the temperature parameter $\tau$.\\
The S$^2$CE objective in Eq.~(\ref{eq:s2ce}) can be rewritten as:
\begin{equation}
    \mathcal L_{S^2CE}=\mathrm{KL}\big(q^m_{\text{sem}}(\cdot|j_u^+)\|p_\theta^m(\cdot|s_u^m)\big)+H(q^m_{\text{sem}}),
\end{equation}
where $H(q^m_{\text{sem}})$ is independent of the model parameters. Therefore, minimizing the S$^2$CE loss is equivalent to:
\begin{equation}
    \arg\min_\theta \mathrm{KL}\big(q^m_{\text{sem}}(\cdot|j_u^+)\,\|\,p_\theta^m(\cdot|s_u^m)\big),
\end{equation}
and at optimality, the distribution satisfies $p_\theta^m(j|s_u^m)\approx q^m_{\text{sem}}(j|j_u^+)$.\\
Now consider the KL divergence between the predictive distributions of the two markets in Eq.~(\ref{eq:kl_ce}). Under S$^2$CE, we obtain:
\begin{equation}
\begin{split}
    \mathrm{KL}(p_\theta^{m_1}\|p_\theta^{m_2}) &= p_\theta^{m_1}(a|s_u^{m_1}) \log \frac{p_\theta^{m_1}(a|s_u^{m_1})}{q_\text{sem}^{m_2}(a|b)} + O(\epsilon)\\
    & \le \log\frac{1}{q_\text{sem}^{m_2}(a|b)} + O(\epsilon).\\
\end{split}
\end{equation}
Ignoring the $O(\epsilon)$ term and substituting the bound in Eq.~(\ref{eq:sem_s2ce}):
\begin{equation}\label{eq:kl_s2ce}
    \mathrm{KL}(p_\theta^{m_1}\|p_\theta^{m_2})\le \log |\mathcal{I}^{m_2}| + \frac{1-\sigma}{\tau}.
\end{equation}
As training in each market approaches optimality, \textit{i.e.}, $\epsilon\to0$. Comparing Eq.~(\ref{eq:kl_ce}) and Eq.~(\ref{eq:kl_s2ce}), we have:
\begin{equation}
    \mathrm{KL}_{S^2CE}(p_\theta^{m_1}\|p_\theta^{m_2})\ll\mathrm{KL}_{CE}(p_\theta^{m_1}\|p_\theta^{m_2}).
\end{equation}

This result shows that S$^2$CE introduces semantic smoothing across markets, emphasizing behavior-level collaboration rather than item-level supervision, which mitigates market heterogeneity and facilitates more effective federated optimization.
\end{proof}

\subsubsection{Global Aggregation} 
When locally pretraining with S$^2$CE, each market focuses on learning behavior-level patterns. After that, the behavior encoders $\{f_\theta^m(\cdot)\}_{m\in\mathcal{M}}$ are uploaded to the server for global aggregation. Formally, the aggregated model is obtained as:
\begin{equation}\label{eq:aggregation}
    f_\theta^{agg}(\cdot) = \sum_{m \in \mathcal{M}} \alpha_m \, f_\theta^m(\cdot),
\end{equation}
where $\alpha_m = 1/|\mathcal{M}|$ denotes the aggregation weight following the widely used FedAvg optimization~\cite{mcmahan2017communication}. In our implementation, the backbone model is SASRec~\cite{kang2018self}, and aggregation is primarily performed on the attention-related parameters, which have been shown to capture the core sequential patterns~\cite{vaswani2017attention, voita2019analyzing, wang2024pre}. After aggregation, the global model $f_\theta^{agg}(\cdot)$ is broadcast back to all markets for the next round of local pretraining.

\subsection{Local Market-specific Fine-tuning}
Through the federated pretraining stage, markets collaboratively learn to predict the semantics of the next item based on the textual modality, yielding a global encoder $f_\theta^{agg}(\cdot)$ that captures market-agnostic behavior-level patterns. Building upon this, we perform parameter-efficient adaptation of $\Theta^m = \{f_\theta^{agg}(\cdot), \mathbf{E}^m\}$ to capture market-specific item-level preferences.

\noindent \textbf{Behavior Encoder.} To enable efficient adaptation while preserving the shared behavioral knowledge, we adopt an efficient strategy by applying low-rank adaptation LoRA~\cite{hu2022lora} to the frozen global encoder. Specifically, the market-specific encoder is defined as $f_\theta^{m}(\cdot) = f_\theta^{agg}(\cdot) + \mathbf W_A \mathbf W_B$, where $\mathbf W_A \in \mathbb{R}^{d \times r}$ and $\mathbf W_B \in \mathbb{R}^{r \times d}$ are trainable matrices with rank $r \ll d$. Each market then performs local fine-tuning using the vanilla CE objective in Eq.~(\ref{eq:ce}).

\noindent \textbf{Item Embeddings.} To further enhance item-level adaptation, we incorporate ID modality~\cite{hou2022towards, xie2022decoupled} on top of unified semantic representations. The enhanced item embedding is defined as:
\begin{equation}\label{eq:idadapter}
\tilde{\mathbf e}_i^m = \mathbf e_i^m + \beta_i \odot g_\theta^m(\mathbf e_i^m),
\end{equation}
where $g_\theta^m(\cdot)$ is a lightweight multi-layer perceptron (MLP) that produces an ID adjustment conditioned on the semantic embedding, and $\beta_i \in (0,1)$ is a learnable gating coefficient controlling the strength of the injected ID information.

Extended theoretical proofs and detailed algorithmic procedures for Section~\ref{sec:methodology} are provided in Appendix~\ref{app:theory} and~\ref{app:algorithm}, respectively.

\begin{table}[thb]\centering
    \caption{Statistics of XMarket across different markets.}
    \resizebox{1.0\columnwidth}{!}{
    \large
    \begin{tabular}{*{9}{c}}
        \toprule
        Markets & ca & mx & us & uk & fr & de & in & jp \\
        \midrule
        \# Users & 5,675 & 1,878 & 35,916 & 4,847 & 2,396 & 2,373 & 239 & 487 \\
        \# Items & 5,772 & 1,645 & 31,125 & 3,392 & 1,911 & 2,210 & 470 & 955 \\
        \# Interactions & 55,045 & 17,095 & 364,339 & 44,515 & 22,905 & 22,247 & 2,015 & 4,485 \\
        Sparsity & 99.83\% & 99.45\% & 99.97\% & 99.73\% & 99.50\% & 99.58\% & 98.21\% & 99.04\% \\
        \bottomrule
    \end{tabular}
    }
    \label{tab:dataset}
\end{table}

\section{EXPERIMENTS}\label{sec:experiments}

\begin{table*}[htbp]
\centering
\caption{Performance comparison. The best and second-best results are bold and underlined, respectively. “Imp.” indicates the improvement of our approach over the strongest baseline, with “*” denoting significance at the 0.05 level (paired $t$-test).}\label{tab:main_exp}
\resizebox{1.0\textwidth}{!}{
\begin{tabular}{llcccccccccccccc}
\toprule
\multirow{2}{*}{\textbf{Datasets}} & \multirow{2}{*}{\textbf{Metrics}}& \multicolumn{3}{c}{\textbf{Local-only Methods}} & \multicolumn{4}{c}{\textbf{Centralized Methods}} & \multicolumn{2}{c}{\textbf{Federated Methods}} & \multicolumn{2}{c}{\textbf{Transfer Methods}} & \multicolumn{1}{c}{\textbf{Ours}} & \multirow{2}{*}{\textbf{Imp.}} \\
\cmidrule(lr){3-5}\cmidrule(lr){6-9}\cmidrule(lr){10-11}\cmidrule(lr){12-13}\cmidrule(lr){14-14}
 &  & \textbf{SASRec} & \textbf{S$^3$Rec} & \textbf{SASRec$_{\text{text}}$} & \textbf{MAML} & \textbf{FOREC} & \textbf{MA} & \textbf{SASRec$_{\text{cen}}$} & \textbf{FedAvg} & \textbf{FedDCSR} & \textbf{UnisRec} & \textbf{CAT-SR} & \textbf{FeCoSR} & \\
\midrule
\multirow{2}*{\textbf{ca}} & HR@10 & 0.1782 & 0.1804 & 0.1763 & 0.1587 & 0.1624 & 0.1706 & 0.1808 & 0.1617 & 0.1761 & 0.1810 & \underline{0.1817} & \textbf{0.1877$^*$} & +3.30\% \\
& NDCG@10 & 0.1132 & 0.1270 & 0.1264 & 0.1086 & 0.1115 & 0.1168 & 0.1253 & 0.1083 & 0.1287 & 0.1163 & \underline{0.1325} & \textbf{0.1399$^*$} & +5.58\% \\
\midrule
\multirow{2}*{\textbf{mx}} & HR@10 & 0.4627 & 0.4915 & 0.5059 & 0.4186 & 0.4275 & 0.4361 & 0.4963 & 0.4015 & 0.4766 & \underline{0.5075} & 0.4989 & \textbf{0.5106} & +0.61\% \\
& NDCG@10 & 0.3351 & 0.3421 & 0.3586 & 0.2947 & 0.3013 & 0.3094 & 0.3236 & 0.2664 & 0.3307 & 0.3056 & \underline{0.3537} & \textbf{0.3656$^*$} & +3.36\% \\
\midrule
\multirow{2}*{\textbf{us}} & HR@10 & 0.1641 & \underline{0.1645} & 0.1421 & 0.1423 & 0.1472 & 0.1518 & 0.1446 & 0.1470 & 0.1400 & 0.1129 & 0.1416 & \textbf{0.1654$^*$} & +0.55\% \\
& NDCG@10 & 0.1171 & \underline{0.1179} & 0.0935 & 0.0997 & 0.1041 & 0.1049 & 0.1056 & 0.0984 & 0.1094 & 0.0617 & 0.0910 & \textbf{0.1191$^*$} & +1.01\% \\
\midrule
\multirow{2}*{\textbf{uk}} & HR@10 & 0.2703 & 0.3055 & 0.2986 & 0.2629 & 0.2695 & 0.2784 & \underline{0.3101} & 0.2706 & 0.2909 & 0.2987 & 0.3013 & \textbf{0.3132$^*$} & +1.00\%\\
& NDCG@10 & 0.1692 & 0.1779 & 0.1991 & 0.1795 & 0.1837 & 0.1908 & 0.2037 & 0.1656 & 0.1938 & 0.1646 & \underline{0.2193} & \textbf{0.2265$^*$} & +3.38\% \\
\midrule
\multirow{2}*{\textbf{fr}} & HR@10 & \textbf{0.2829} & 0.2715 & 0.2742 & 0.2314 & 0.2369 & 0.2453 & 0.2775 & 0.2454 & 0.2748 & 0.2671 & 0.2748 & \underline{0.2824} & - \\
& NDCG@10 & 0.1919 & 0.1903 & 0.1883 & 0.1689 & 0.1724 & 0.1796 & 0.1981 & 0.1669 & 0.1971 & 0.1411 & \underline{0.2024} & \textbf{0.2028} & +0.20\% \\
\midrule
\multirow{2}*{\textbf{de}} & HR@10 & 0.2928 & 0.2885 & 0.2896 & 0.2143 & 0.2217 & 0.2315 & \textbf{0.3025} & 0.2377 & 0.2593 & 0.2474 & 0.2907 & \underline{0.2944} & - \\
& NDCG@10 & 0.1901 & 0.1919 & 0.2013 & 0.1506 & 0.1568 & 0.1643 & 0.2126 & 0.1749 & 0.1986 & 0.1342 & \underline{0.2140} & \textbf{0.2204$^*$} & +2.99\% \\
\midrule
\multirow{2}*{\textbf{in}} & HR@10 & 0.4435 & 0.4310 & 0.4477 & 0.3925 & 0.4018 & 0.4187 & 0.4477 & 0.3933 & 0.4351 & 0.3849 & \underline{0.4561} & \textbf{0.4728$^*$} & +3.66\% \\
& NDCG@10 & 0.2517 & 0.2726 & \underline{0.3037} & 0.2314 & 0.2386 & 0.2478 & 0.2547 & 0.1781 & 0.2662 & 0.1929 & 0.2840 & \textbf{0.3562$^*$} & +17.28\% \\
\midrule
\multirow{2}*{\textbf{jp}} & HR@10 & 0.2895 & 0.2854 & 0.2752 & 0.1816 & 0.1882 & 0.1956 & 0.2669 & 0.1971 & 0.2526 & 0.2156 & \underline{0.2936} & \textbf{0.3060$^*$} & +4.22\% \\
& NDCG@10 & \underline{0.1851} & 0.1775 & 0.1658 & 0.1218 & 0.1267 & 0.1321 & 0.1811 & 0.1162 & 0.1501 & 0.1106 & 0.1781 & \textbf{0.2062$^*$} & +11.40\% \\
\bottomrule
\end{tabular}
}
\end{table*}

\subsection{Experimental Setup}
\subsubsection{Datasets} We conduct experiments on XMarket\footnote{https://xmrec.github.io}, a publicly available real-world cross-market recommendation dataset collected from Amazon. The dataset contains user–item interactions together with rich item textual descriptions. It covers eight electronic markets distributed across three continents and has been widely adopted in previous CMR studies~\cite{bonab2021cross, bhargav2023market, wang2024pre}. The characteristics of the datasets are summarized in Table~\ref{tab:dataset}.

\subsubsection{Compared Methods}
To comprehensively evaluate the effectiveness of our proposed method, we compare it with several representative baselines from four categories:\\
\textbf{i) Single-market local training methods} train models independently on each market without cross-market collaboration. \textit{SASRec}~\cite{kang2018self} is a Transformer-based sequential recommendation (SR) model. \textit{S\text{$^3$}-Rec}~\cite{zhou2020s3} is a self-supervised SR model that leverages multiple pretraining tasks. \textit{SASRec$_{\text{text}}$} is a variant of SASRec where item embeddings are derived from the textual modality.\\
\textbf{ii) Cross-market centralized collaborative methods} jointly train models using data from multiple markets in a centralized manner.
\textit{MAML}~\cite{kang2023outlier} is a meta-learning framework that learns transferable models across markets. \textit{FOREC}~\cite{bonab2021cross} partially freezes the MAML model and performs fine-tuning on the target market. \textit{MA}~\cite{bhargav2023market} leverages market embeddings and auxiliary market data to enhance cross-market recommendation. \textit{SASRec$_{\text{cen}}$} is a centralized variant of SASRec trained on the combined data from all markets.\\
\textbf{iii) Cross-market federated collaborative methods} optimize models across all clients under the federated paradigm.
\textit{FedAvg}~\cite{mcmahan2017communication} is a classical federated algorithm that averages model parameters across clients. \textit{FedDCSR}~\cite{zhang2024feddcsr} addresses feature misalignment in cross-domain sequential recommendation via feature decoupling.\\
\textbf{iv) One-to-one transfer-based methods} pretrain a model on a source market and adapt it to target markets. \textit{UniSRec}~\cite{hou2022towards} is a universal SR model that learns transferable representations for downstream adaptation. \textit{CAT-SR}~\cite{wang2024pre} mitigates item popularity shift in the source market and transfers the pretrained model for market-specific fine-tuning.

\subsubsection{Evaluation Protocol} To measure recommendation quality, we adopt two widely used metrics, Hit Rate (HR@$N$) and Normalized Discounted Cumulative Gain (NDCG@$N$), where $N=10$. Following previous works~\cite{hou2022towards, wang2024pre}, we employ the leave-one-out evaluation protocol, where the last interaction of each user is used for testing, the second-last for validation, and the remaining interactions for training. To ensure accurate and stable evaluation, we rank the ground-truth item against all items that the user has not interacted with, \textit{i.e.}, full ranking.


\subsubsection{Implementation Details} We implement our method using the popular open-source library RecBole\footnote{https://recbole.io}~\cite{zhao2021recbole}. To ensure a fair comparison, we adopt the same basic configurations for all RecBole-based methods, including SASRec, S$^3$Rec, UniSRec, CAT-SR, and our proposed FeCoSR. Specifically, we set the batch size to 256,  optimize with Adam optimizer, and search the learning rate in $\{0.01, 0.001, 0.0001\}$. Other hyperparameters follow the optimal settings recommended in the original papers. For MAML, FOREC, and MA, we reproduce their results using the publicly available code\footnote{https://github.com/samarthbhargav/efficient-xmrec} and reported settings. For all federated methods, we uniformly set the number of global rounds to 20. For FedAvg and FedDCSR, the remaining settings are kept consistent with their original papers. For our method, the temperature parameter $\tau$ is tuned within $[0.01, 0.1]$ with a step size of 0.01 during pretraining, and the low-rank parameter $r$ is selected from $\{1,2,4,8,16\}$ during fine-tuning.

\subsection{Overall Performance}
We evaluate FeCoSR against various baselines across eight markets. 
From the results in Table~\ref{tab:main_exp}, we have the following observations:\\
\textbf{i) FeCoSR consistently outperforms all baseline methods}, achieving the best performance on the majority of markets in terms of HR@10 and all markets in terms of NDCG@10. This demonstrates the strong effectiveness of our framework for CMR. We note that none of the methods, including FeCoSR, show clear improvements over SASRec in terms of HR@10 on \textbf{fr}. This may be attributed to the strong inherent preference patterns in this market, making it difficult to benefit from cross-market collaboration. In the market \textbf{de}, FeCoSR performs slightly worse than the centralized SASRec$_{\text{cen}}$ on HR@10. This is expected since centralized methods have access to all data, which may be more beneficial. Nevertheless, FeCoSR still outperforms other privacy-preserving methods.\\
\textbf{ii) Our method effectively mitigates source degradation}. Previous privacy-preserving CMR approaches typically adopt a one-to-one transfer paradigm, which weakens market-specific preferences in the source market to facilitate transferability. However, this process leads to irreversible performance degradation in the source market, even after fine-tuning. For example, on the market \textbf{us}, UniSRec and CAT-SR perform significantly worse than locally trained SASRec. In contrast, our method enables collaborative pretraining across markets via federated optimization, allowing all markets to benefit without sacrificing their own performance.\\
\textbf{iii) Our method effectively avoids negative transfer} through its two-stage design, consistently improving performance over single-market models such as SASRec and S$^3$Rec. The federated pretraining stage focuses on learning behavior-level patterns, reducing conflicts caused by item-level preference heterogeneity across markets, and promoting knowledge sharing. The local fine-tuning stage further adapts the model to market-specific preferences, achieving a better balance between generalization and localization.\\
\textbf{iv) FeCoSR enables effective and robust collaboration under privacy constraints.} 
Centralized training fails to reconcile heterogeneous item preferences across markets while also violating privacy. Compared with SASRec, SASRec$_{\text{cen}}$ even hurts performance on some markets, \textit{e.g.}, \textbf{us}, \textbf{fr}, \textbf{jp}. Existing federated methods, \textit{e.g.}, FedAvg, FedDCSR, preserve privacy but do not explicitly address market heterogeneity. In contrast, FeCoSR explicitly addresses this issue via the S$^2$CE objective, achieving more stable performance improvements across markets.\\
\textbf{v) FeCoSR achieves substantial improvements on data-scarce markets} such as \textbf{in} and \textbf{jp}, demonstrating its strong capability to leverage cross-market knowledge and effectively address the data sparsity issue in CMR scenarios.

\begin{table}[t]
\centering
\caption{Ablation study of FeCoSR. The best results are bold.}\label{tab:ablation}
\resizebox{1.0\columnwidth}{!}{
\begin{tabular}{l|cccccccc}
\toprule
\textbf{HR@10} & ca & mx & us & uk & fr & de & in & jp \\
\midrule
FeCoSR & \textbf{0.1877} & \textbf{0.5106} & \textbf{0.1654} & \textbf{0.3132} & \textbf{0.2824} & \textbf{0.2944} & \textbf{0.4728} & \textbf{0.3060} \\
w/o Fed	& 0.1836 & 0.5037 & 0.1612 & 0.3049 & 0.2775 & 0.2885 & 0.4644 & 0.2854\\
w/o S$^2$CE & 0.1793 & 0.4920 & 0.1506 & 0.2983 & 0.2628 & 0.2804 & 0.4435 & 0.2752 \\
w/o LR & 0.1851 & 0.5080 & 0.1641 & 0.3058 & 0.2764 & 0.2836 & 0.4770 & 0.2854 \\
w/o ID & 0.1838 & 0.5037 & 0.1615 & 0.3019 & 0.2677 & 0.2890 & 0.4603 & 0.2895 \\
\toprule
\textbf{NDCG@10} &ca & mx & us & uk & fr & de & in & jp \\
\midrule
FeCoSR & \textbf{0.1399} & \textbf{0.3656} & \textbf{0.1191} & \textbf{0.2265} & \textbf{0.2028} & \textbf{0.2204} & \textbf{0.3562} & \textbf{0.2062} \\
w/o Fed & 0.1368 & 0.3498 & 0.1154 & 0.2167 & 0.1977 & 0.2074 & 0.2903 & 0.1966\\
w/o S$^2$CE & 0.1348 & 0.3064 & 0.1005 & 0.2103 & 0.1829 & 0.2056 & 0.2446 & 0.1891 \\
w/o LR & 0.1382 & 0.3630 & 0.1189 & 0.2169 & 0.1876 & 0.2116 & 0.3180 & 0.1972 \\
w/o ID & 0.1359 & 0.3469 & 0.1142 & 0.2032 & 0.1743 & 0.1905 & 0.2259 & 0.1746 \\
\bottomrule
\end{tabular}
}
\end{table}

\subsection{Ablation Study}
We conduct an ablation study by removing key designs of FeCoSR to examine their individual impacts. For the pretraining stage, \textbf{Fed} denotes the collaborative optimization via federated aggregation across markets, and \textbf{S$^2$CE} refers to the semantic-aware objective designed to mitigate cross-market heterogeneity. For the fine-tuning stage, \textbf{LR} denotes the low-rank adaptation applied to the sequence encoder, and \textbf{ID} represents the incorporation of ID information into item embeddings. From the results in Table~\ref{tab:ablation}, we observe that:\\
i) Removing the federated collaboration (w/o Fed) leads to consistent performance degradation across all markets, indicating that knowledge sharing is essential for improving overall performance.\\
ii) Replacing S$^2$CE with the vanilla CE objective (w/o S$^2$CE) results in the most significant performance drop across all markets, especially in HR@10. This confirms that S$^2$CE effectively mitigates cross-market heterogeneity by transforming heterogeneous item-level preferences into shared behavior-level patterns, thereby facilitating more effective collaboration.\\
iii) Removing item-level adaptation from either the sequence encoder (w/o LR) or the item embeddings (w/o ID) results in noticeable performance degradation. The impact is particularly evident on ranking-sensitive metrics such as NDCG@10, highlighting the importance of both components for accurate item-level modeling.

Overall, FeCoSR achieves superior performance by jointly leveraging cross-market collaboration, behavior-level modeling, and item-level refinement.

\begin{figure}[t]
\centering
\includegraphics[width=1.0\columnwidth]{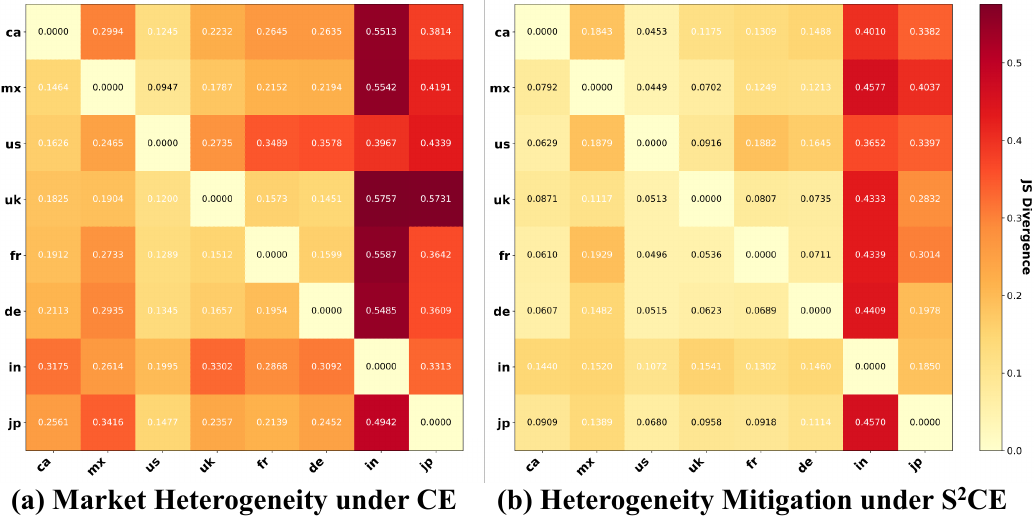}
\caption{Inter-market heterogeneity under CE and S$^2$CE.}
\label{pic:heterogeneity}
\end{figure}

\subsection{In-depth Analysis}
\subsubsection{Heterogeneity Mitigation}\label{sec:heterogeneity} To further analyze inter-market heterogeneity, we measure the Jensen–Shannon (JS) divergence between model predictions across markets. As shown in Figure~\ref{pic:heterogeneity}, each row represents a market $m$ providing behavioral data $\mathcal{S}^m$, and each column represents a market $m’$ providing the behavior encoder $f_\theta^{m'}(\cdot)$. For each entry, we compute the JS divergence between the overall item prediction distribution of $f_\theta^{m'}(\cdot)$ and that of the local encoder $f_\theta^{m}(\cdot)$. A larger JS divergence indicates that the behavior patterns learned by market $m’$ are less compatible with those of market $m$. Figure~\ref{pic:heterogeneity}(a) presents the JS divergence matrix under the vanilla CE. We observe large divergences between many market pairs, indicating significant heterogeneity across markets. In contrast, Figure~\ref{pic:heterogeneity}(b) shows the results under the proposed S$^2$CE, where the divergences are consistently reduced, suggesting improved cross-market collaboration. Notably, the markets \textbf{in} and \textbf{jp} exhibit larger divergences from others, which may be attributed to regional differences or their relatively limited training data, leading to less sharable behavioral modeling. Nevertheless, our method effectively mitigates their adverse impact on other markets.

\begin{figure}[t]
\centering
\includegraphics[width=1.0\columnwidth]{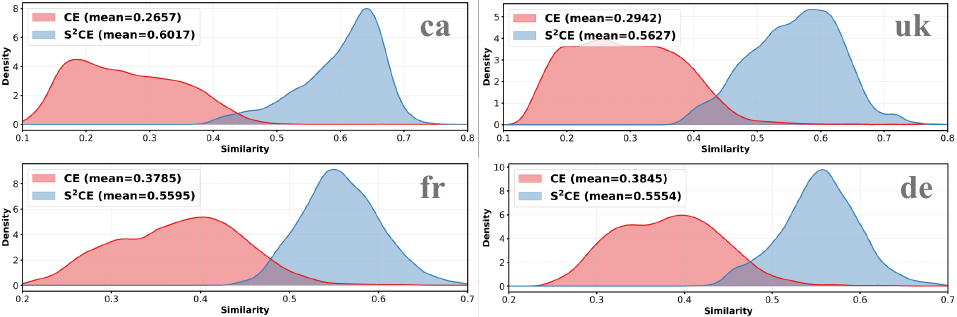}
\caption{Distributions of inter-user representation similarities under behavior encoders trained with CE and S$^2$CE.}
\label{pic:behaviorlevel}
\end{figure}

\subsubsection{Behavior-level Patterns} We compute inter-user representation similarities and visualize their distributions using kernel density estimation (KDE)~\cite{chen2017tutorial}. As some markets shown in Figure~\ref{pic:behaviorlevel}, user similarities learned by the vanilla CE are generally lower and exhibit a relatively flat distribution, indicating that user representations are dispersed in the latent space and focus on item-level preferences. In contrast, the S$^2$CE-based model leads to an overall increase in user similarity, suggesting the emergence of clustered structures in the representation space. The users within the same cluster share similar semantic interests in items, reflecting underlying behavioral regularities. Importantly, the distribution under S$^2$CE still maintains a reasonable spread rather than collapsing into a narrow peak, indicating that the model captures informative knowledge. Overall, these results demonstrate that S$^2$CE effectively captures behavior-level patterns from item semantics, rather than focusing on item-level preferences that are often market-specific.

\subsubsection{Item-level Preferences} We further examine the distribution of inter-user representation similarities using KDE. The comparison between the pretrained and fine-tuned behavior encoders is shown in Figure~\ref{pic:itemlevel} (above). Compared with the pretrained encoder, the overall similarity decreases after fine-tuning, indicating that user representations become less clustered. This suggests a transition from capturing shared semantic signals to modeling fine-grained interests, \textit{i.e.}, the model refines behavior-level patterns towards item-level preferences. We also visualize the item embedding distributions before and after incorporating ID modality, as shown in Figure~\ref{pic:itemlevel} (below). The textual modality provides a shared semantic space that captures general item semantics and enables cross-market behavior modeling. Building upon this shared space, each market performs localized adjustments by injecting ID signals, resulting in market-specific item-level preferences while retaining global behavioral knowledge.

\begin{figure}[t]
\centering
\includegraphics[width=1.0\columnwidth]{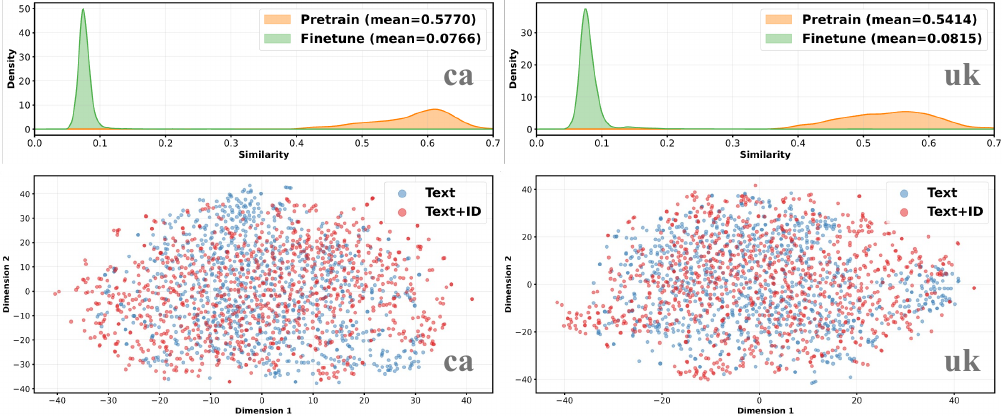}
\caption{(Above) Distributions of inter-user representation similarities before and after behavior encoder fine-tuning. (Below) Visualization of item embeddings before and after incorporating ID modality.}
\label{pic:itemlevel}
\end{figure}

\begin{figure}[t]
\centering
\includegraphics[width=0.99\columnwidth]{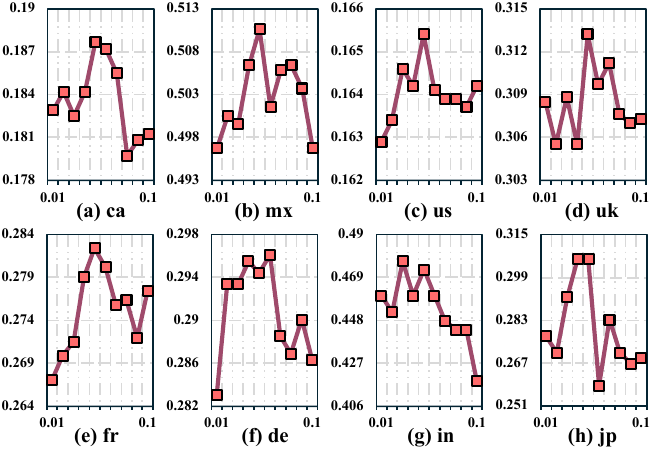}
\caption{Effect of the temperature parameter $\tau$ during federated pretraining, evaluated by HR@10.}
\label{pic:param_tau_HR}
\end{figure}

\subsection{Hyperparameter Analysis}\label{sec:exp_hyper}
We analyze the impact of the two main hyperparameters as follows:\\
\textbf{i) The temperature coefficient $\bm{\tau}$} during federated pretraining. For simplicity, we use a uniform $\tau$ value for all markets during pretraining. Figure~\ref{pic:param_tau_HR} shows the impact of $\tau$ on HR@10 across different markets. Most markets achieve their best performance within the range $\tau=0.03$–$0.07$, which is consistent with the role of $\tau$, \textit{i.e.}, controlling the smoothness of the semantic soft supervision. Specifically, a too-small $\tau$ approaches one-hot supervision similar to vanilla CE, exacerbating inter-market heterogeneity and hindering collaboration. Conversely, a too-large $\tau$ over-smooths the supervision, causing negative samples to receive non-negligible predicted probabilities and thereby reducing the model’s discriminative ability. Additionally, we observe that some markets, \textit{e.g.}, \textbf{in} and \textbf{jp}, show notable drops at extreme $\tau$ values, which can be explained by their limited local training data and greater reliance on cross-market knowledge. Considering all markets, we choose $\tau=0.05$ as a balanced setting for our experiments. \\
\textbf{ii) The low-rank parameter $\bm{r}$} during local fine-tuning. The impact of $r$ on HR@10 across different markets is shown in Figure~\ref{pic:param_rank_HR}. When $r=0$, \textit{i.e.}, without fine-tuning, the performance is consistently worse, while small values of $r$ achieve the best results, indicating that lightweight adaptation is sufficient to capture market-specific item-level preferences. As $r$ increases, the performance first improves and then degrades, suggesting that excessive adaptation may lead to overfitting or disrupt the shared behavior-level patterns learned during federated pretraining. Overall, relatively small ranks, \textit{e.g.}, $r\in\{1,2,4\}$, provide a good trade-off across all markets.

Detailed experimental settings and additional results for Section~\ref{sec:experiments} are provided in Appendix~\ref{app:experiments}.

\begin{figure}[t]
\centering
\includegraphics[width=0.98\columnwidth]{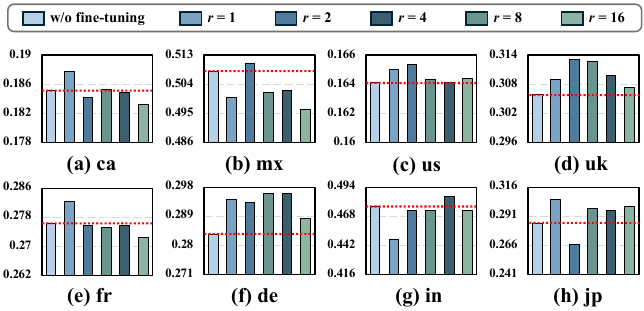}
\caption{Effect of the low-rank $r$ during local fine-tuning, evaluated by HR@10.}
\label{pic:param_rank_HR}
\end{figure}

\section{CONCLUSION}
In this work, we revisit CMR and identify two fundamental limitations of existing one-to-one transfer paradigms, \textit{i.e.}, source degradation and negative transfer. To this end, we propose FeCoSR, a novel federated collaboration framework that enables many-to-many knowledge sharing across markets. FeCoSR introduces a two-stage design that decouples behavior-level collaboration and item-level preference modeling, allowing all markets to benefit without sacrificing their own performance. Furthermore, we reveal that the vanilla CE objective exacerbates market heterogeneity and hinders effective federated optimization. Then, we propose the S$^2$CE objective, which leverages semantic soft supervision to promote cross-market behavioral sharing. Extensive experiments on real-world datasets demonstrate that FeCoSR consistently outperforms other methods while preserving privacy.

\bibliographystyle{ACM-Reference-Format}
\bibliography{sample-base}

\appendix

\section{Detailed Theoretical Analysis}\label{app:theory}

\begin{remark}[Cross-Market Shared Behavior and Market-Specific Items]
In cross-market recommendation, user behaviors often exhibit shared high-level patterns across markets, while item-level preferences are typically market-specific due to variations in item availability, cultural context, and local interests. 
Such patterns have been observed in multiple multi-domain datasets. For instance, in Amazon product review datasets~\cite{mcauley2015image}, users across different categories often show similar sequential or co-purchase behavior, even though the specific items differ by category. This observation motivates our analysis of how different training objectives balance shared behavior-level learning and market-specific item-level adaptation.
\end{remark}

Based on the above observation, we provide a theoretical analysis of \emph{Market Heterogeneity under Vanilla CE} and \emph{Heterogeneity Mitigation under S$^2$CE}. Compared to the main text, we further investigate how CE and S$^2$CE respectively influence federated optimization through heterogeneity across markets. For clarity, the formal notations used throughout the paper are summarized in Table~\ref{tab:notation}.

\begin{table}[t]
\centering
\caption{Notation Table.}\label{tab:notation}
\resizebox{1.0\columnwidth}{!}{
\begin{tabular}{cc}
\toprule
\textbf{Notation} & \textbf{Meaning} \\
\midrule

$\mathcal{M}$, $m$ & set of markets (clients), index of market \\
$\mathcal{U}^m$, $\mathcal{I}^m$, $u$ & set of users and items, index of user \\
$|\mathcal{U}^m|$, $|\mathcal{I}^m|$, $d$ & number of items, item embedding dimension\\
$\mathcal{S}^m$, $s_u^m$ & dataset of market $m$, interaction sequence of user $u$\\
$i_k^m$, $j_u^+$ & the $k$-th interacted item, ground-truth next item of user $u$\\
$t_i$, $w_k$ & textual content of item $i$, the $k$-th word\\
$\Theta^m=\{\mathbf{E}^m, f_\theta^m(\cdot)\}$ & model parameters for market $m$ \\
$\mathbf{E}^m, \mathbf{e}_i^m$ & item embedding matrix,  embedding of item $i$\\
$f_\theta^m(\cdot)$, $\mathbf{h}_u^m$& sequential behavior encoder, user representation\\
$\mathbf{r}_u^m$ & predicted scores over items \\

$\mathbf{W}_A\mathbf{W}_B$ & low-rank matrices for behavior encoder\\
$g_\theta^m(\cdot)$ & a lightweight multi-layer perceptron (MLP) for item embeddings\\

$\Theta_{\text{pre}}^{\text{fed}}$, $\alpha_m$ & federated pretrained model, aggregation weight of market $m$ \\

$\mathcal{L}_{\text{pre}}$, $\mathcal{L}_{\text{ft}}$ & pretraining loss, fine-tuning loss \\
$\mathcal{L}_{\text{CE}}$, $\mathcal{L}_{\text{S}^2\text{CE}}$ & cross-entropy loss, semantic soft cross-entropy loss\\
$sim(\cdot\,,\cdot)$, $\tau$ & similarity function (cosine similarity), temperature parameter \\
$p_\theta^m(j \mid \mathbf{s}_u^m)$ & predicted probability of item $j$ \\
$q_{\text{sem}}^m(j \mid j_u^+)$ & semantic soft weight over item $j$ \\
$\text{KL}(\cdot \| \cdot)$, $\text{JS}(\cdot\,, \cdot)$ & Kullback--Leibler divergence, Jensen–Shannon divergence \\
$H(\cdot, \cdot)$, $H(\cdot)$ & cross-entropy and entropy function\\
$\epsilon$, $\sigma$ & a very small positive constant, semantic similarity \\
$\Theta_m^\varepsilon$, $\varepsilon$ & the set of $\varepsilon$-optimal solutions for market $m$\\

$R$, $E$, $B$, $\eta$ & global rounds, local
epochs, batch size, learning rate \\

\bottomrule
\end{tabular}
}
\end{table}

\subsection{Market Heterogeneity under Vanilla CE}

\begin{proposition}[Market Heterogeneity under Vanilla CE]
In cross-market recommendation, shared behavior-level patterns may exist across markets. However, local training with the vanilla CE objective drives each market toward near one-hot predictions over item-level preferences. This induces amplified heterogeneity across markets and disturbs the collaboration of federated optimization.
\end{proposition}

\begin{proof}
\textbf{Step 1: One-Hot Predictive Distributions under CE.}

Under the vanilla CE objective, the model maximizes the likelihood of the ground-truth item:
\begin{equation}
    \min_\theta -\log p_\theta^m(j_u^+ | s_u^m)
    \;\Longleftrightarrow\;
    \max_\theta p_\theta^m(j_u^+ | s_u^m).
\end{equation}

As the model approaches optimality, the predictive distribution becomes highly concentrated on the ground-truth item. 
Consider two users from markets $m_1$ and $m_2$ with ground-truth items $a$ and $b$. Then, we have:
\begin{equation}
p_\theta^{m_1}(a | s_u^{m_1}) \ge 1 - \epsilon,
\quad
p_\theta^{m_2}(b | s_u^{m_2}) \ge 1 - \epsilon,
\end{equation}
where $\epsilon \to 0$, and:
\begin{equation}
\sum_{j \neq a} p_\theta^{m_1}(j | s_u^{m_1}) \le \epsilon,
\quad
p_\theta^{m_2}(a | s_u^{m_2}) \le \epsilon.
\end{equation}

\noindent \textbf{Step 2: Market Heterogeneity by Predictive Divergence.}

\noindent The KL divergence between the two predictive distributions is:
\begin{equation}
\begin{split}
    & \mathrm{KL}\big(p_\theta^{m_1}(\cdot | s_u^{m_1}) \,\|\, p_\theta^{m_2}(\cdot | s_u^{m_2})\big) = \sum_j p_\theta^{m_1}(j | s_u^{m_1}) \log \frac{p_\theta^{m_1}(j | s_u^{m_1})}{p_\theta^{m_2}(j | s_u^{m_2})}\\
    &= p_\theta^{m_1}(a | s_u^{m_1})\log \frac{p_\theta^{m_1}(a | s_u^{m_1})}{p_\theta^{m_2}(a | s_u^{m_2})} + \sum_{j \neq a} p_\theta^{m_1}(j | s_u^{m_1})\log \frac{p_\theta^{m_1}(j | s_u^{m_1})}{p_\theta^{m_2}(j | s_u^{m_2})}.
\end{split}
\end{equation}

\noindent Using the concentration bounds in Step 1:
\begin{equation}
p_\theta^{m_1}(a | s_u^{m_1})
\log \frac{p_\theta^{m_1}(a|s_u^{m_1})}{p_\theta^{m_2}(a | s_u^{m_2})}
\ge (1 - \epsilon)\log \frac{1 - \epsilon}{\epsilon},
\end{equation}
while the remaining terms are bounded by $O(\epsilon)$. Hence:
\begin{equation}
\mathrm{KL}(p_\theta^{m_1}\|p_\theta^{m_2}) \ge (1 - \epsilon)\log \frac{1 - \epsilon}{\epsilon} + O(\epsilon)
\approx \log(1/\epsilon),
\end{equation}
which diverges as $\epsilon \to 0$. This divergence indicates that predictive distributions induced by different markets become mutually incompatible, \textit{market heterogeneity}.

\noindent \textbf{Step 3: Heterogeneity-Disrupted Federated Optimization.}

\noindent Define the $\varepsilon$-optimal set for each market:
\begin{equation}
\Theta_m^\varepsilon 
= \big\{ \theta : \mathcal{L}_{\mathrm{CE}}(\theta; \mathcal{S}^m) 
\le \inf_{\theta'} \mathcal{L}_{\mathrm{CE}}(\theta'; \mathcal{S}^m) + \varepsilon \big\}.
\end{equation}

\noindent From Step 1, for any $\theta \in \mathcal{S}_m^\varepsilon$, there exists $\delta(\varepsilon) \to 0$ such that:
\begin{equation}
p_\theta^m(j_u^+ | s_u^m) \ge 1 - \delta(\varepsilon),
\end{equation}
\textit{i.e.}, the predictive distribution is highly concentrated on the ground-truth item. Now consider two markets $m_1 \neq m_2$ with different ground-truth items $a \neq b$.  
For any $\theta \in \Theta_{m_1}^\varepsilon$, the CE loss on market $m_1$ is near its minimum:
\begin{equation}
\mathcal{L}_{\mathrm{CE}}(\theta; \mathcal{S}^{m_1}) 
\approx -\log p_\theta^{m_1}(a | s_u^{m_1}) 
\le -\log(1 - \delta(\varepsilon)) + \varepsilon,
\end{equation}
while the loss on market $m_2$ is lower bounded as:
\begin{equation}
\mathcal{L}_{\mathrm{CE}}(\theta; \mathcal{S}^{m_2})
\ge -\log p_\theta^{m_2}(b | s_u^{m_2})
\ge -\log (1 - \delta(\varepsilon)) + c,
\end{equation}
for some constant $c > 0$, since the model concentrated on item $a$ in $m_1$ cannot simultaneously assign high probability to $b$ in $m_2$.

By symmetry, the same argument holds if we start from $\theta \in \Theta_{m_2}^\varepsilon$.  
Hence, for sufficiently small $\varepsilon$, the $\varepsilon$-optimal sets of the two markets do not intersect:
\begin{equation}
\Theta_{m_1}^\varepsilon \cap \Theta_{m_2}^\varepsilon = \varnothing,
\end{equation}
and more generally, we have:
\begin{equation}
\bigcap_{m \in \mathcal{M}} \Theta_{m}^\varepsilon = \varnothing.
\end{equation}
Since federated learning seeks a single global parameter that minimizes a weighted sum of local CE losses:
\begin{equation}
\theta^{\mathrm{fed}}_{\mathrm{CE}} 
= \arg\min_\theta \sum_{m \in \mathcal{M}} \alpha_m \, \mathcal{L}_{\mathrm{CE}}(\theta; \mathcal{S}^m),
\end{equation}
the disjointness of the $\varepsilon$-optimal sets implies that there does not exist a parameter $\theta$ that is simultaneously $\varepsilon$-optimal for all markets. This formally shows that market heterogeneity induced by the CE objective disturbs federated optimization.

\end{proof}

\subsection{Heterogeneity Mitigation under S$^2$CE}

\begin{proposition}[Heterogeneity Mitigation under S$^2$CE]\label{pro:s2ce}
The S$^2$CE loss encourages the model to tolerate items that are semantically similar to the ground-truth. This supervision promotes the learning of shared behavior-level patterns and effectively mitigates the heterogeneity caused by item-level preferences across markets.
\end{proposition}

\begin{proof}
\textbf{Step 1: Soft Predictive Distributions under S$^2$CE.}

The S$^2$CE loss is defined as the cross-entropy between the soft target distribution $q_{\text{sem}}(\cdot|j_u^+)$ and the model prediction $p_\theta^m(\cdot|s_u^m)$:
\begin{equation}
\mathcal L_{\mathrm{S^2CE}}(\theta; \mathcal{S}^m)
=  - \sum_{j} q_{\text{sem}}(j|j_u^+) \log p_\theta^m(j|s_u^m).
\end{equation}

\noindent \noindent We start from the definition of cross-entropy between two distributions $q(\cdot)$ and $p(\cdot)$:
\begin{equation}
H(q, p) = - \sum_j q(j) \log p(j).
\end{equation}
We add and subtract $\log q(j)$ inside the summation:
\begin{equation}
\begin{split}
H(q, p) 
&= - \sum_j q(j) \log p(j)\\
& = - \sum_j q(j) \log \frac{p(j)}{q(j)} - \sum_j q(j) \log q(j)\\
&= \sum_j q(j) \log \frac{q(j)}{p(j)} 
- \sum_j q(j) \log q(j).
\end{split}
\end{equation}
Recognizing the two terms:
\begin{equation}
\mathrm{KL}(q \| p) = \sum_j q(j) \log \frac{q(j)}{p(j)}, 
\quad
H(q) = - \sum_j q(j) \log q(j),
\end{equation}
we finally arrive at:
\begin{equation}
H(q, p) = \mathrm{KL}(q \| p) + H(q).
\end{equation}
Then, we can rewrite the loss as:
\begin{equation}
\mathcal L_{\mathrm{S^2CE}}(\theta; \mathcal{S}^m)
= 
\mathrm{KL}\big(q^m_{\text{sem}}(\cdot|j_u^+) \,\|\, p_\theta^m(\cdot|s_u^m)\big)
+ H\big(q^m_{\text{sem}}(\cdot|j_u^+)\big).
\end{equation}
Since the entropy term $H(q^m_{\text{sem}})$ does not depend on the model parameters, minimizing $\mathcal L_{\mathrm{S^2CE}}$ is equivalent to:
\begin{equation}
\arg\min_\theta 
\mathrm{KL}\big(q^m_{\text{sem}}(\cdot|j_u^+) \,\|\, p_\theta^m(\cdot|s_u^m)\big),
\end{equation}
so that at optimality, we have:
\begin{equation}
p_\theta^m(j|s_u^m) \approx q^m_\text{sem}(j|j_u^+),
\end{equation}
\textit{i.e.}, the predictive distribution is softened across semantically similar items, not concentrated on a single item.

\noindent \textbf{Step 2: Bounded Cross-Market Divergence.}

\noindent From Step 1, we have:
\begin{equation}
p_\theta^{m_1}(a|s_u^{m_1}) \ge 1 - \delta(\varepsilon), \quad
p_\theta^{m_2}(a|s_u^{m_2}) \approx q_\text{sem}^{m_2}(a|b),
\end{equation}
where $\delta(\varepsilon)\to0$ as $\varepsilon\to0$.

\noindent Denote the semantic similarity between items $a$ and $b$ as $\sigma = sim(\mathbf e_a, \mathbf e_b)\in[-1,1]$.  
Under the S$^2$CE loss, the soft target distribution for item $a$ in market $m_2$ with ground-truth $b$ is:
\begin{equation}
\begin{split}
    q_\text{sem}^{m_2}(a|b) & = \frac{\exp(sim(\mathbf e_a, \mathbf e_b)/\tau)}{\sum_{k\in \mathcal I^{m_2}} \exp(sim(\mathbf e_k, \mathbf e_b)/\tau)} \\ & = \frac{\exp(sim(\mathbf e_a,\mathbf e_b) / \tau)}{\sum_{k \neq b} \exp(sim(\mathbf{e}_k^\top \mathbf{e}_b) / \tau) + \exp(sim(\mathbf{e}_b^\top \mathbf{e}_b) / \tau)} \\ &\ge \frac{\exp(\sigma / \tau)}{|\mathcal{I}^{m_2}| \cdot \exp(1 / \tau)} = \frac{\exp({\sigma-1}/{\tau})}{|\mathcal{I}^{m_2}|}.
\end{split}
\end{equation}
The KL divergence between predictive distributions of the two markets satisfies:
\begin{equation}
\begin{split}
\mathrm{KL}(p_\theta^{m_1}\|p_\theta^{m_2}) 
&= \sum_j p_\theta^{m_1}(j|s_u^{m_1}) \log \frac{p_\theta^{m_1}(j|s_u^{m_1})}{p_\theta^{m_2}(j|s_u^{m_2})} \\
&= p_\theta^{m_1}(a|s_u^{m_1}) \log \frac{p_\theta^{m_1}(a|s_u^{m_1})}{p_\theta^{m_2}(a|s_u^{m_2})} + O(\epsilon) \\
&\approx  p_\theta^{m_1}(a|s_u^{m_1}) \log \frac{p_\theta^{m_1}(a|s_u^{m_1})}{q_\text{sem}^{m_2}(a|b)} + O(\epsilon) \\
&\le \log \frac{1}{q_\text{sem}^{m_2}(a|b)} + O(\epsilon) \\
&\le \log |\mathcal I^{m_2}| + \frac{1-\sigma}{\tau} + O(\epsilon).
\end{split}
\end{equation}
Hence, unlike vanilla CE where KL diverges as $\epsilon\to0$, S$^2$CE ensures a bounded divergence between markets:
\begin{equation}
\mathrm{KL}_{S^2CE}(p_\theta^{m_1}\|p_\theta^{m_2}) \ll \mathrm{KL}_{CE}(p_\theta^{m_1}\|p_\theta^{m_2}),
\end{equation}
indicating that the predictive distributions are now compatible and market heterogeneity is mitigated.

\noindent \textbf{Step 3: Heterogeneity Mitigation in Federated Optimization.}

\noindent Consider two markets $m_1 \neq m_2$ with ground-truth items $a \neq b$.  
\noindent From the definition of the semantic target, for any item $j \in \mathcal{I}^m$, we have:
\begin{equation}
q_{\text{sem}}^{m}(j|j_u^+) 
= \frac{\exp(sim(\mathbf e_j, \mathbf e_{j_u^+})/\tau)}{\sum_{k \in \mathcal I^m} \exp(sim(\mathbf e_k, \mathbf e_{j_u^+})/\tau)} > 0.
\end{equation}
Hence, the semantic target distributions $q_{\text{sem}}^{m_1}(\cdot|a)$ and $q_{\text{sem}}^{m_2}(\cdot|b)$ assign strictly positive probability mass to all items, implying that they have fully overlapping support. Moreover, for any item $j$, since $sim(\mathbf e_j, \mathbf e_{j_u^+}) \in [-1,1]$, we obtain:
\begin{equation}
\frac{\exp(-1/\tau)}{|\mathcal I^m| \cdot \exp(1/\tau)}
\le q_{\text{sem}}^{m}(j|j_u^+) 
\le \frac{\exp(1/\tau)}{|\mathcal I^m| \cdot \exp(-1/\tau)}.
\end{equation}
Thus, the ratio between two semantic distributions is bounded:
\begin{equation}
\frac{q_{\text{sem}}^{m_1}(j|a)}{q_{\text{sem}}^{m_2}(j|b)} \le C_0,
\end{equation}
for some constant $C_0 > 0$ depending only on $\tau$ and the item set size.

\noindent Therefore, the KL divergence between the two semantic distributions is bounded:
\begin{equation}
\mathrm{KL}\big(q_{\text{sem}}^{m_1} \,\|\, q_{\text{sem}}^{m_2}\big)
= \sum_j q_{\text{sem}}^{m_1}(j|a) \log \frac{q_{\text{sem}}^{m_1}(j|a)}{q_{\text{sem}}^{m_2}(j|b)}
\le \log C_0,
\end{equation}
for some constant $C = \log C_0 > 0$ depending on $\sigma$ and $\tau$.

\noindent Define the $\varepsilon$-optimal set for market $m$ under S$^2$CE:
\begin{equation}
\Theta_{m}^\varepsilon 
= \big\{ \theta : \mathcal{L}_{\mathrm{S^2CE}}(\theta; \mathcal{S}^m) 
\le \inf_{\theta'} \mathcal{L}_{\mathrm{S^2CE}}(\theta'; \mathcal{S}^m) + \varepsilon \big\}.
\end{equation}

\noindent From Step 1, minimizing S$^2$CE is equivalent to minimizing:
\begin{equation}
\mathrm{KL}\big(q^m_{\text{sem}}(\cdot|j_u^+) \,\|\, p_\theta^m(\cdot|s_u^m)\big),
\end{equation}
thus for any $\theta \in \mathcal{S}_m^\varepsilon$, we have:
\begin{equation}
\mathrm{KL}\big(q^m_{\text{sem}}(\cdot|j_u^+) \,\|\, p_\theta^m(\cdot|s_u^m)\big) \le \varepsilon.
\end{equation}

\noindent For any $\theta \in \Theta_{m_1}^\varepsilon$, combining the above and using the smoothness of the semantic distributions, we obtain:
\begin{equation}
\mathrm{KL}(q_{\text{sem}}^{m_2} \| p_\theta^{m_1}) \le \mathrm{KL}\big(q_{\text{sem}}^{m_2} \,\|\, q_{\text{sem}}^{m_1}\big)
+ \mathrm{KL}\big(q_{\text{sem}}^{m_1} \,\|\, p_\theta^{m_1}\big) \le C + \varepsilon.
\end{equation}

\noindent This shows that the predictive distribution learned from market $m_1$ remains close to the semantic target distribution of market $m_2$, \textit{i.e.}, cross-market supervision becomes compatible up to a bounded discrepancy. Combined with the bounded divergence in Step 2, this implies that $p_\theta^{m_2}$ does not deviate arbitrarily from $p_\theta^{m_1}$, and thus:
\begin{equation}
\mathrm{KL}(q_{\text{sem}}^{m_2} \| p_\theta^{m_2}) \le C' + \varepsilon,
\end{equation}
for some constant $C' > 0$. Recall that S$^2$CE loss can be written as:
\begin{equation}
\mathcal{L}_{\mathrm{S^2CE}}(\theta; \mathcal{S}^{m})
= \mathrm{KL}(q_{\text{sem}}^{m} \| p_\theta^{m}) + H(q_{\text{sem}}^{m}),
\end{equation}
where $H(q_{\text{sem}}^{m})$ is independent of $\theta$. For any $\theta \in \mathcal{S}_{m_1}^\varepsilon$, we have:
\begin{equation}
\mathcal{L}_{\mathrm{S^2CE}}(\theta; \mathcal{S}^{m_1})
\le \inf_{\theta'} \mathcal{L}_{\mathrm{S^2CE}}(\theta'; \mathcal{S}^{m_1}) + \varepsilon,
\end{equation}
and combining with the above bound, we obtain:
\begin{equation}
\mathcal{L}_{\mathrm{S^2CE}}(\theta; \mathcal{S}^{m_2})
\le \inf_{\theta'} \mathcal{L}_{\mathrm{S^2CE}}(\theta'; \mathcal{S}^{m_2}) + C' + \varepsilon.
\end{equation}

\noindent This shows that a parameter $\theta$ that is $\varepsilon$-optimal for market $m_1$ is also near-optimal for market $m_2$ up to a bounded gap. By symmetry, the same holds when exchanging $m_1$ and $m_2$. Hence, for sufficiently large tolerance, we have:
\begin{equation}
\Theta_{m_1}^\varepsilon \cap \Theta_{m_2}^{\varepsilon + C'} \neq \varnothing,
\quad
\bigcap_{m \in \mathcal{M}} \Theta_{m}^{\varepsilon + C'} \neq \varnothing.
\end{equation}

\noindent Consequently, the federated objective under S$^2$CE:
\begin{equation}
\theta^{\mathrm{fed}}_{\mathrm{S^2CE}} 
= \arg\min_\theta \sum_{m \in \mathcal{M}} \alpha_m \, \mathcal{L}_{\mathrm{S^2CE}}(\theta; \mathcal{S}^m),
\end{equation}
admits solutions that lie in the approximate $\varepsilon$-optimal sets of multiple markets, enabling near-optimal performance across markets up to a bounded gap. This formally shows that S$^2$CE mitigates market heterogeneity, thereby facilitating more effective federated optimization.

\end{proof}

\section{Detailed Algorithm Procedure}\label{app:algorithm}
In this section, we provide a detailed description of the proposed training pipeline, including the federated collaborative pretraining stage, \textit{i.e.}, Algorithm~\ref{alg:pretrain}, and the local market-specific fine-tuning stage, \textit{i.e.}, Algorithm~\ref{alg:finetune}.

\textbf{Federated Collaborative Pretraining.} We first perform federated pretraining to learn a shared behavior encoder across markets while preserving data privacy. Specifically, the server initializes a unified item semantic embedding matrix $\mathbf{E}$ according to Eq.~(\ref{eq:initialization}), and distributes market-specific subsets $\mathbf{E}^m$ to each market $m$. The global encoder $f_\theta^{agg(0)}(\cdot)$ is also initialized and shared.

At each global round $r$, each market first performs local pretraining. Specifically, it downloads the global behavior encoder $f_\theta^{agg(r-1)}(\cdot)$ from the previous round for local initialization. After that, it optimizes the semantic soft objective $\mathcal{L}_{S^2CE}$ on its private dataset $\mathcal{S}^m = \{s_u^m\}_{u \in \mathcal{U}^m}$ for local updates. The updated local models $\{f_\theta^{m(r)}(\cdot)\}_{m \in \mathcal{M}}$ are subsequently aggregated on the server to obtain the updated global model $f_\theta^{agg(r)}(\cdot)$, which is then broadcast to all markets for the next round of local training. This process repeats for $R$ rounds, yielding the final pretrained encoder $f_\theta^{agg(R)}(\cdot)$.

Note that throughout the entire pretraining stage, the item embeddings $\mathbf{E}^m$ remain frozen, and only $f_\theta^m(\cdot)$ is updated to capture behavior-level patterns.

\begin{algorithm}[tb]
\small
\caption{Federated Collaborative Pretraining}
\label{alg:pretrain}
\raggedright
\textbf{Input:} markets $\mathcal{M}$, local datasets $\{\mathcal{S}^m\}_{m \in \mathcal{M}}$, global rounds $R$, local epochs $E$, batch size $B$, learning rate $\eta$.\\
\textbf{Output:} pretrained models $\{f_\theta^{agg}(\cdot),\mathbf{E}^m\}_{m\in\mathcal{M}}$.\\

\textbf{Server executes:} \\
\begin{algorithmic}[1]
\STATE Generate unified item semantic embeddings $\mathbf{E}$ according to Eq.~(\ref{eq:initialization}).
\STATE Initialize the global behavior encoder $f_\theta^{agg(0)}(\cdot)$.

\FOR{each market $m \in \mathcal{M}$}
    \STATE Initialize and freeze item embeddings $\mathbf{E}^m = \mathbf{E}[\mathcal{I}^m]$.
    \STATE Initialize encoder $f_\theta^{m(0)}(\cdot) \gets f_\theta^{agg(0)}(\cdot)$.
\ENDFOR

\FOR{each global round $r = 1,2,\cdots,R$}
    \FOR{each market $m \in \mathcal{M}$ \textbf{in parallel}}
        \STATE $f_{\theta}^{m(r)}(\cdot) \gets \text{LocalPretrain}(f_\theta^{agg(r-1)}(\cdot), \mathcal{S}^m)$.
    \ENDFOR
    \STATE $f_{\theta}^{agg(r)} \gets$ Aggregate $\{f_{\theta}^{m(r)}\}_{m \in \mathcal{M}}$ according to Eq.~(\ref{eq:aggregation}).
    \STATE Broadcast aggregated $f_{\theta}^{agg(r)}(\cdot)$ to all markets.
\ENDFOR

\RETURN $\{f_\theta^{agg(R)}(\cdot),\mathbf{E}^m\}_{m\in\mathcal{M}}$
\end{algorithmic}

\textbf{LocalPretrain}$(f_\theta^{agg(r-1)}(\cdot), \mathcal{S}_m)$: \\
\begin{algorithmic}[1]
\STATE Initialize $f_\theta^{m(r-1)}(\cdot) \gets f_\theta^{agg(r-1)}(\cdot)$.
\FOR{each local epoch $e = 1,2,\cdots,E$}
    \FOR{each batch $b \in \mathcal{S}_m$}
        \STATE Minimize the semantic soft cross-entropy loss $\mathcal{L}_{S^2CE}$ in Eq.~(\ref{eq:s2ce}).
        \STATE Update $f_\theta^{m(r-1)}(\cdot)$ using $\eta \nabla \mathcal{L}_{S^2CE}$.
    \ENDFOR
\ENDFOR
\RETURN $f_\theta^{m(r)}(\cdot)$
\end{algorithmic}
\end{algorithm}

\begin{algorithm}[tb]
\small
\caption{Local Market-specific Fine-tuning}
\label{alg:finetune}
\raggedright
\textbf{Input:} markets $\mathcal{M}$, local datasets $\{\mathcal{S}^m\}_{m \in \mathcal{M}}$, pretrained models $\{f_\theta^{agg}(\cdot),\mathbf{E}^m\}_{m\in\mathcal{M}}$, local epochs $E$, batch size $B$, learning rate $\eta$.\\
\textbf{Output:} fine-tuned models $\{f_\theta^m(\cdot), \tilde{\mathbf{E}}^m\}_{m \in \mathcal{M}}$.\\

\begin{algorithmic}[1]
\FOR{each market $m \in \mathcal{M}$ \textbf{in parallel}}
    \STATE Freeze $f_\theta^{agg}(\cdot)$ and apply trainable LoRA parameters $\mathbf{W}_A, \mathbf{W}_B$.
    \STATE Freeze $\mathbf{E}^m$ and apply a trainable ID encoder $g_\theta^m(\cdot)$.
    \FOR{each local epoch $e = 1,2,\cdots,E$}
        \FOR{each batch $b \in \mathcal{S}_m$}
            \STATE Enhance item embeddings to obtain $\tilde{\mathbf{E}}^m$ via Eq.~(\ref{eq:idadapter}).
            \STATE Compute the vanilla cross-entropy loss $\mathcal{L}_{CE}$ in Eq.~(\ref{eq:ce}).
            \STATE Update $\mathbf{W}_A, \mathbf{W}_B$ and $g_\theta^m(\cdot)$ using $\eta \nabla \mathcal{L}_{CE}$.
        \ENDFOR
    \ENDFOR
\ENDFOR
\RETURN $\{f_\theta^m(\cdot), \tilde{\mathbf{E}}^m\}_{m \in \mathcal{M}}$
\end{algorithmic}
\end{algorithm}

\textbf{Local Market-specific Fine-tuning.} During local fine-tuning, the model shifts from capturing shared behavior-level patterns to focusing on market-specific item-level preferences. For the pretrained encoder $f_\theta^{agg}(\cdot)$, we freeze its parameters and introduce lightweight trainable LoRA parameters $\mathbf{W}_A, \mathbf{W}_B$ for efficient adaptation. For the item embeddings $\mathbf{E}^m$, we still keep them frozen and apply a trainable ID encoder $g_\theta^m(\cdot)$ to enhance item representations.

During this process, each market optimizes the vanilla cross-entropy loss $\mathcal{L}_{CE}$ over its local dataset $\mathcal{S}^m$ for accurate next-item prediction, and only the LoRA parameters $\mathbf{W}_A\mathbf{W}_B$ and the ID encoder $g_\theta^m(\cdot)$ are updated.

Overall, the proposed two-stage framework achieves a balance between cross-market knowledge sharing and market-specific adaptation, leading to improved performance and robustness across heterogeneous markets.

\section{Additional Experiment Results}\label{app:experiments}
\subsection{Heterogeneity Mitigation}

\begin{figure}[ht]
\centering
\includegraphics[width=1.0\columnwidth]{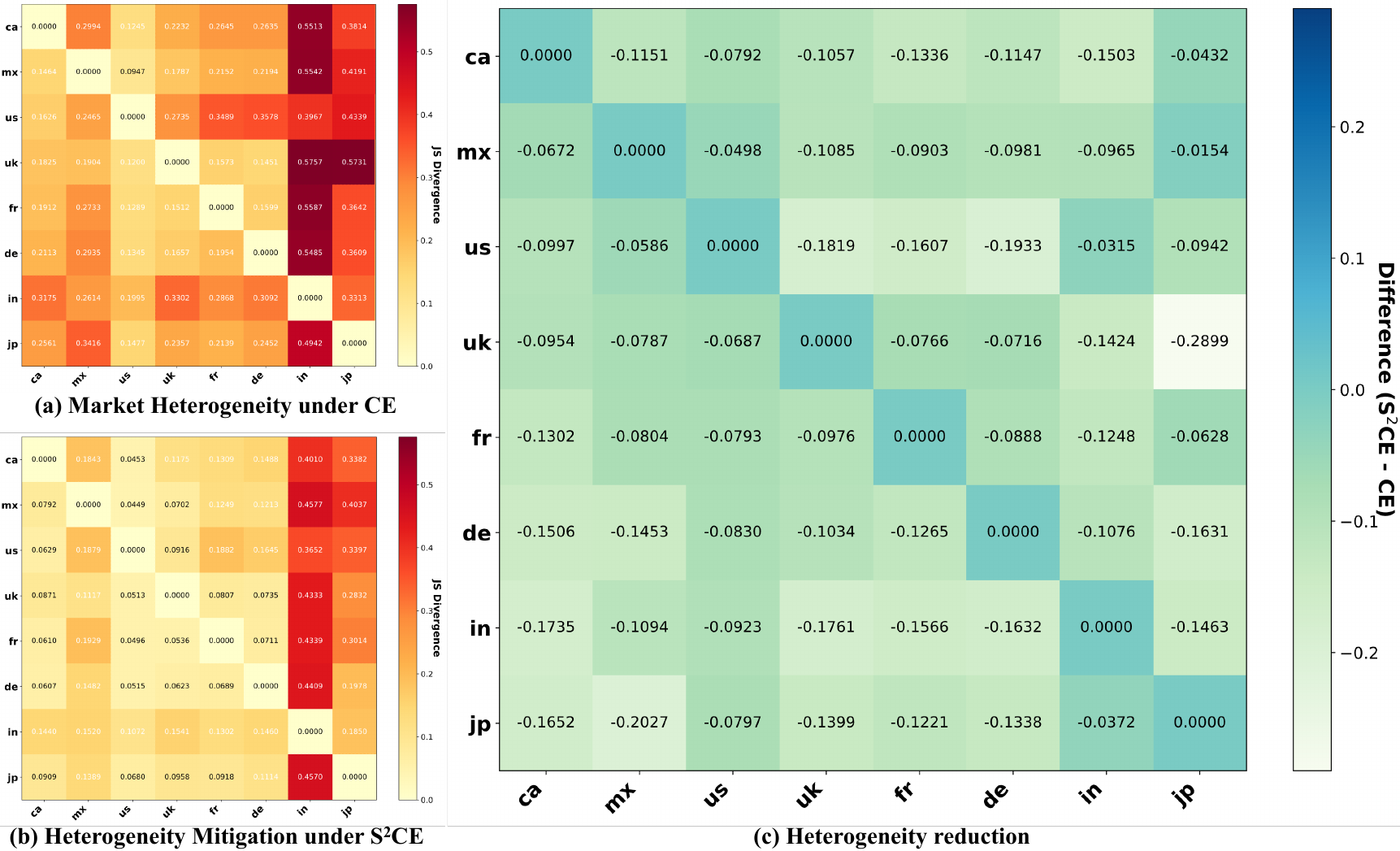}
\caption{Inter-market heterogeneity under CE and S$^2$CE.}
\label{pic:heterogeneity_reduced}
\end{figure}

\subsubsection{Inter-market Heterogeneity}

To quantify inter-market heterogeneity, we measure the divergence between prediction distributions induced by behavior encoders from different markets. 
Specifically, for each market $m \in \mathcal{M}$, we use its local dataset $\mathcal{S}^m = \{s_u^m\}_{u \in \mathcal{U}^m}$ as evaluation data. 
Given a behavior encoder $f_\theta^{m'}(\cdot)$ from another market $m'$, we compute the predicted item distribution for each sequence $s_u^m$ using the instance-level predictive distribution $p_\theta^{m'}(i | s_u^m)$. The overall prediction distribution is obtained by averaging over all users:
\begin{equation}
P^{m' \rightarrow m}(i) = \frac{1}{|\mathcal{U}^m|} \sum_{u \in \mathcal{U}^m} p_\theta^{m'}(i | s_u^m),
\end{equation}
which measures how well the behavioral patterns learned from another market $m'$ generalize to the local data of market $m$ in terms of the induced prediction distribution. Meanwhile, this formulation extends the instance-level predictive distribution $p_\theta^m(\cdot | s_u^m)$ used in the theoretical analysis to a population-level distribution over all local items, providing a quantitative basis for measuring inter-market heterogeneity.

To measure the heterogeneity between two markets $m$ and $m'$, we compute the Jensen--Shannon (JS) divergence between their prediction distributions:
\begin{equation}
\mathrm{JS}(P^{m \rightarrow m}, P^{m' \rightarrow m}) 
= \frac{1}{2} \mathrm{KL}(P^{m \rightarrow m} \,\|\, M) 
+ \frac{1}{2} \mathrm{KL}(P^{m' \rightarrow m} \,\|\, M),
\end{equation}
where $M = \frac{1}{2}(P^{m \rightarrow m} + P^{m' \rightarrow m})$ and $\mathrm{KL}(\cdot\|\cdot)$ denotes the Kullback--Leibler divergence. 
Here, $P^{m \rightarrow m}$ captures the intrinsic preference distribution of market $m$, while $P^{m' \rightarrow m}$ represents the prediction distribution induced by applying the behavior encoder from market $m'$ to the data of market $m$. The JS divergence then quantifies the alignment between the two distributions. A larger JS divergence indicates stronger heterogeneity between markets.

\subsubsection{Results Analysis} Figure~\ref{pic:heterogeneity_reduced} compares inter-market heterogeneity under the CE and S$^2$CE objectives, and also presents a difference matrix highlighting the reduction achieved by S$^2$CE. In the matrices, each row corresponds to the local data from a given market $m$, and each column corresponds to the behavior encoder from another market $m'$. 

This empirical observation is consistent with the theoretical analysis. Under the CE objective, behavioral patterns learned from other markets exhibit limited transferability to the local data, resulting in large divergence values and thus strong inter-market heterogeneity, as shown in Figure~\ref{pic:heterogeneity_reduced}(a). 
In contrast, S$^2$CE significantly reduces the divergence across markets, as shown in Figure~\ref{pic:heterogeneity_reduced}(b). This improvement can be attributed to the shift from item-level discrimination to behavior-level pattern learning. By encouraging the model to capture shared semantic structures rather than market-specific item preferences, S$^2$CE alleviates the heterogeneity across markets. The difference matrix in Figure~\ref{pic:heterogeneity_reduced}(c) further provides empirical evidence for the effectiveness of S$^2$CE.

\begin{figure}[ht]
\centering
\includegraphics[width=1.0\columnwidth]{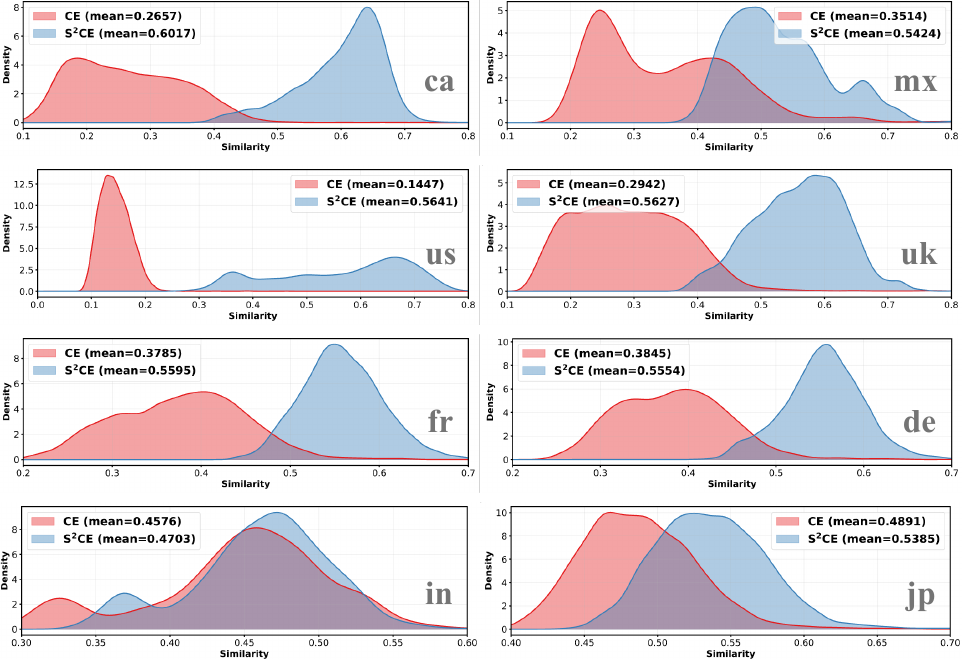}
\caption{Distributions of inter-user representation similarities under behavior encoders trained with CE and S$^2$CE.}
\label{pic:CEvsS2CE}
\end{figure}

\begin{figure*}[htbp]
\centering
\includegraphics[width=1.0\textwidth]{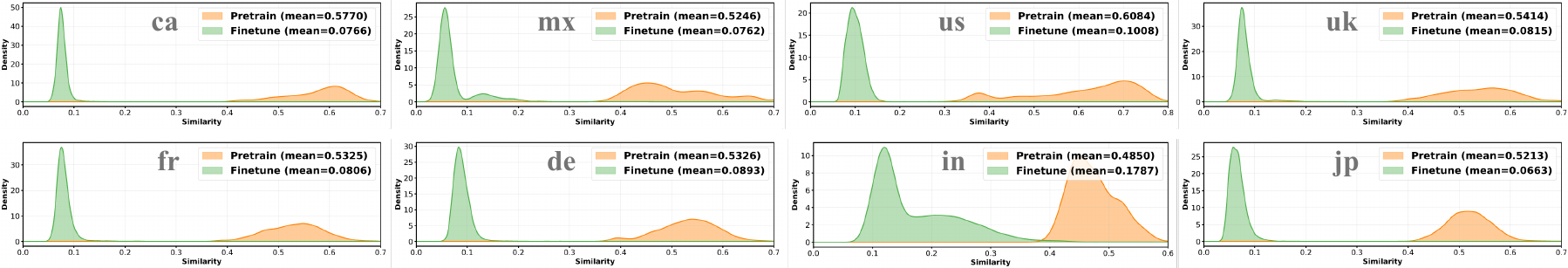}
\caption{Distributions of inter-user representation similarities before and after behavior encoder fine-tuning.}
\label{pic:behavior}
\end{figure*}

\begin{figure*}[htbp]
\centering
\includegraphics[width=1.0\textwidth]{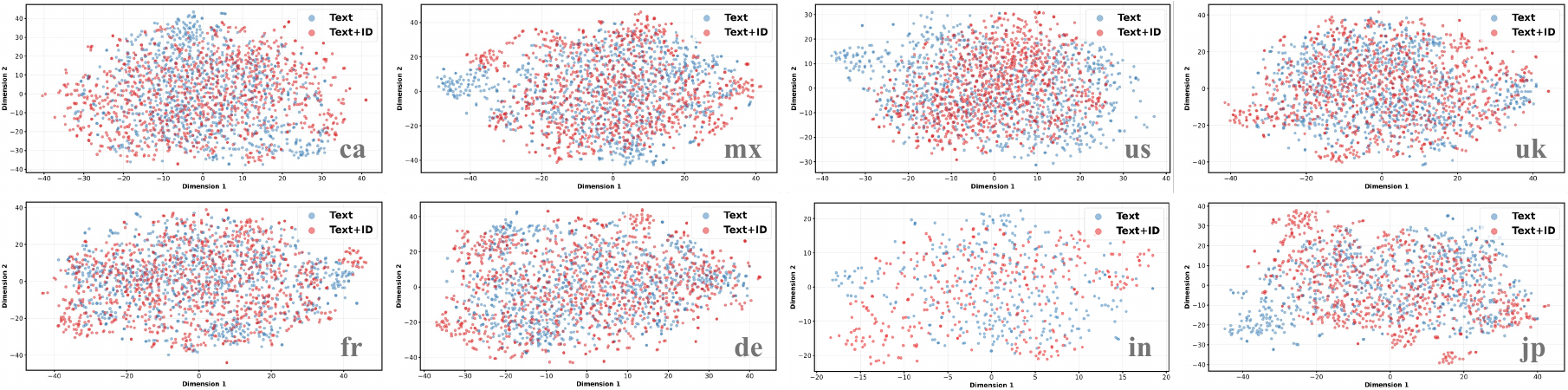}
\caption{Visualization of item embeddings before and after incorporating the ID modality.}
\label{pic:embedding}
\end{figure*}

\subsection{Behavior-level Patterns}
\subsubsection{Inter-user Representation Similarity} Let $\mathbf{h}_u^m = f_\theta^m(s_u^m)$ denote the representation of user $u$ in market $m$. For any two users $u, v \in \mathcal{U}^m$, we define their similarity based on the $\ell_2$ distance as:
\begin{equation}
\text{sim}(u,v) = \frac{1}{1 + \|\mathbf{h}_u^m - \mathbf{h}_v^m\|_2},
\end{equation}
which maps distances to the range $(0,1]$, where larger values indicate higher similarity. We adopt this formulation instead of cosine similarity, as cosine similarity in high-dimensional spaces is often less discriminative, making it harder to distinguish variations in user representations.

Let the similarities across all user pairs, \textit{i.e.}, 
$\{sim(u,v) | u,v \in \mathcal{U}^m, u \neq v\}$, 
be denoted by $\{s_i\}_{i=1}^N$ for brevity, where 
$N = \frac{|\mathcal{U}^m| \left(|\mathcal{U}^m| - 1\right)}{2}$. 
To examine their distribution, we apply kernel density estimation (KDE), which provides a smooth estimate of the underlying density:
\begin{equation}
\hat{p}^m(s) = \frac{1}{N} \sum_{i=1}^{N} K\left(\frac{s - s_i}{h}\right),
\end{equation}
where $K(\cdot)$ denotes a kernel function (e.g., Gaussian), and $h$ is the bandwidth parameter. 
This enables us to characterize the dispersion of user representations in the latent space.

\subsubsection{Results Analysis} We compare the inter-user representation similarity learned under the CE and S$^2$CE objectives, as shown in Fig.~\ref{pic:CEvsS2CE}. Under the CE objective, the similarity values are concentrated in a relatively low range, indicating that user representations are more dispersed and primarily capture individual item-level preferences. suggesting the emergence of clustered structures, where users within the same cluster share similar semantic interests over items. This implies that S$^2$CE encourages the model to capture more holistic behavior-level patterns. Importantly, the similarity distribution under S$^2$CE still maintains sufficient spread, indicating that the model learns informative behavioral patterns without collapsing to a trivial solution.

We further analyze two representative markets, \texttt{us} and \texttt{in}, corresponding to the largest and smallest datasets, respectively. 
For \texttt{us}, CE-based similarities are mainly concentrated in $[0.1, 0.2]$ with a sharp distribution, whereas S$^2$CE leads to a wider distribution range. This is likely due to the larger number of users, which allows for more diverse similarity patterns. For \texttt{in}, the increase in similarity under S$^2$CE is less pronounced. 
A possible reason is that the limited number of users restricts the formation of shared patterns. Overall, S$^2$CE yields consistently higher similarity values than CE across different markets, which overall supports the aforementioned observations and analysis.

\subsection{Item-level Preferences}

\subsubsection{Experimental Setup} We experimentally validate the two key components of the sequential recommendation model, \textit{i.e.}, the behavior encoder and the item embeddings. For the behavior encoder, we analyze the distribution of inter-user representation similarity to compare the model before and after applying low-rank fine-tuning to $f_\theta^m(\cdot)$. For the item embeddings, we employ t-SNE to visualize the embeddings before and after applying $g_\theta^m(\cdot)$ to incorporate the ID modality, \textit{i.e.}, $\mathbf{E}^m$ and $\tilde{\mathbf{E}}^m$. For clearer visualization, we randomly sample up to 400 items.

\subsubsection{Results Analysis} 

Figure~\ref{pic:behavior} shows the distribution of inter-user representation similarity before and after fine-tuning the behavior encoder. A consistent trend is observed across all markets, that is, after fine-tuning, the similarity values are significantly reduced and become concentrated in the range of $[0, 0.2]$. This indicates that user representations become more discriminative, focusing more on item-level distinctions rather than the shared behavior-level patterns observed before fine-tuning. Moreover, since we adopt low-rank fine-tuning, most pretrained weights are preserved. This suggests that the underlying behavioral patterns are largely retained, while the model shifts toward capturing more market-specific characteristics.

Figure~\ref{pic:embedding} visualizes the item embeddings before and after incorporating the ID modality. During pretraining, we initialize item embeddings using text modality to obtain a stable representation space, facilitating collaborative learning of behavioral patterns. Based on this, incorporating the ID modality further refines the item embeddings. It can be observed that the updated embeddings remain in a similar representation space while exhibiting clearer structure, indicating that the pretrained information is largely preserved while enabling item-level adaptation to market-specific characteristics.

\subsection{Hyperparameter Analysis}
We further provide additional experimental results in Section~\ref{sec:exp_hyper} to analyze the impact of hyperparameters, reported by NDCG@10.

\subsubsection{The temperature coefficient ${\tau}$} As shown in Fig.~\ref{pic:param_tau_NDCG}, the observed trend is consistent with that of HR@10. 
Specifically, although different datasets exhibit varying sensitivity to $\tau$, they all achieve near-optimal performance within a moderate range, \textit{i.e.}, $\tau \in [0.03, 0.07]$. This reflects the trade-off induced by $\tau$, \textit{i.e.}, small $\tau$ yields a highly peaked distribution with one-hot-like supervision, whereas large $\tau$ results in over-smoothing and reduced discriminability. In our experiments, we adopt a shared $\tau$ across all markets for training consistency. In practice, allowing each market to explore its own $\tau$ may further improve performance. However, it would require more sophisticated tuning strategies to maintain consistency and effective coordination across markets. Nevertheless, using a unified setting of $\tau = 0.05$ still consistently outperforms baseline methods, demonstrating both the effectiveness of our approach and its robustness to the choice of hyperparameters.

\begin{figure}[h]
\centering
\includegraphics[width=1.0\columnwidth]{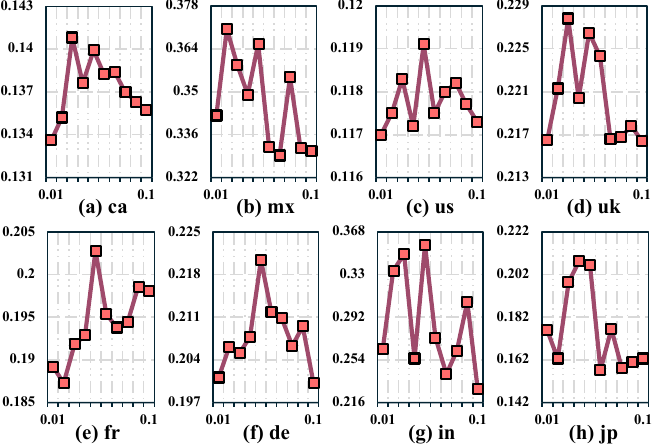}
\caption{Effect of the temperature parameter $\tau$ during federated pretraining, evaluated by NDCG@10.}
\label{pic:param_tau_NDCG}
\end{figure}

\begin{figure}[h]
\centering
\includegraphics[width=1.0\columnwidth]{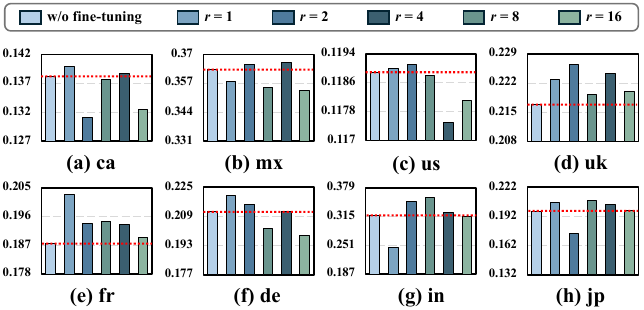}
\caption{Effect of the low-rank $r$ during local fine-tuning, evaluated by NDCG@10.}
\label{pic:param_rank_NDCG}
\end{figure}

\subsubsection{The low-rank parameter ${r}$} As the NDCG@10 results shown in Table~\ref{pic:param_rank_NDCG}, the best performance is achieved after fine-tuning. 
Specifically, smaller rank values, \textit{e.g.}, $r\in\{1,2,4\}$, tend to yield better performance across most markets. For example, $r=2$ or $4$ achieves the best or near-best results on several datasets such as \textbf{mx}, \textbf{uk}, and \textbf{in}, while larger ranks, \textit{e.g.}, 16, do not bring further improvements and may even slightly degrade performance. These results suggest that the pretrained model has already captured sufficiently rich behavioral patterns, and only lightweight item-level adaptation is needed to fit local market characteristics.

\end{document}